\documentclass[aps,showpacs,twocolumn,superscriptaddress]{revtex4-2}

\usepackage{graphicx} 
\usepackage{epsfig,amsmath,amssymb,amsfonts,color,algorithm,algcompatible,amsthm,bm,braket,dcolumn,makecell,multirow,xcolor}
\usepackage{natbib}
\usepackage{comment}
\usepackage{array}
\usepackage{hhline}

\usepackage{algorithm}
\usepackage{algpseudocode}

\newtheorem{lemma}{Lemma}
\newtheorem{theorem}{Theorem}
\renewcommand{\qedsymbol}{$\blacksquare$}

\usepackage[colorlinks=true]{hyperref}
\hypersetup{
           breaklinks=true,   
           colorlinks=true,   
           pdfusetitle=true, 
           citecolor=blue,
           urlcolor=blue
           }

\begin{document}
\preprint{APS/123-QED}

\title{Near-Heisenberg-limited parallel amplitude estimation \\with logarithmic depth circuit}

\author{Kohei~Oshio}
\email{kohei.oshio@mizuho-rt.co.jp}
\affiliation{Mizuho Research \& Technologies, Ltd., 2-3 Kanda-Nishikicho, Chiyoda-ku, Tokyo, 101-8443, Japan}
\affiliation{Quantum Computing Center, Keio University, 3-14-1 Hiyoshi, Kohoku-ku, Yokohama, Kanagawa, 223-8522, Japan}

\author{Kaito~Wada}
\email{wkai1013keio840@keio.jp}
\thanks{K.O. and K.W. contributed equally to this work.}
\affiliation{Graduate School of Science and Technology, Keio University, 3-14-1 Hiyoshi, Kohoku-ku, Yokohama, Kanagawa, 223-8522, Japan}

\author{Naoki~Yamamoto}
\email{yamamoto@appi.keio.ac.jp}
\affiliation{Quantum Computing Center, Keio University, 3-14-1 Hiyoshi, Kohoku-ku, Yokohama, Kanagawa, 223-8522, Japan}
\affiliation{Department of Applied Physics and Physico-Informatics, Keio University, 3-14-1 Hiyoshi, Kohoku-ku, Yokohama, Kanagawa, 223-8522, Japan}


\begin{abstract}

Quantum amplitude estimation is one of the core subroutines in quantum algorithms. 
This paper gives a parallelized amplitude estimation (PAE) algorithm that simultaneously achieves near-Heisenberg scaling in the total number of queries and sub-linear scaling in the circuit depth, with respect to the estimation precision.
The algorithm is composed of a global GHZ state followed by separated low-depth Grover circuits optimized by quantum signal processing techniques; the number of qubits in the GHZ state and the depth of each circuit is tunable as a trade-off way, which particularly enables even near-Heisenberg-limited and logarithmic-depth algorithm for amplitude estimation.
We prove that this trade-off scaling is nearly optimal with use of the parallel quantum adversary method, against folklore on the impossibility of efficient parallelization in amplitude estimation. 
The proposed algorithm has a form of distributed quantum computing, which may be suitable for device implementation. 

\end{abstract}

\maketitle

\textit{Introduction.}---
Estimating unknown parameters in quantum systems is a central topic in quantum metrology~\cite{giovannetti2006quantum, giovannetti2011advances}. 
Many efficient estimation strategies have been developed in various settings; in particular, two major strategies to quantum-limited estimation are the \textit{parallel} and \textit{sequential strategies}, which roughly speaking, utilize large entanglement and long coherence time, respectively. 
The techniques in quantum metrology are powerful, and there has been growing interest in applying such techniques to the development of efficient algorithms for quantum computation scenario \cite{knill2007optimal, kimmel2015robust, wang2021minimizing, dutkiewicz2022heisenberg, ding2023even, ni2023low, wada2024quantum, oshio2024adaptive, wada2025heisenberg}.

In those estimation algorithms, Quantum Amplitude Estimation (QAE)~\cite{brassard2002quantum} is an essential component.
Because it can be applied to expectation value estimation for any observable, it has numerous applications such as chemistry~\cite{kassal2008polynomial, lin2020near, dong2022ground,PhysRevResearch.4.043210}, finance~\cite{rebentrost2018quantum, woerner2019quantum, stamatopoulos2020option}, and machine learning~\cite{wiebe2015quantum, wiebe2016quantum, kapoor2016quantum, kerenidis2019q}. 
Specifically, in QAE, we are given an $n$-qubit ($n \ge 2$) unitary operator $U_a$ (and $U_a^\dagger$) that encodes the target parameter $a \in [0, 1]$ as
\begin{equation}\label{eq:quantum_oracle_qae}
    U_a\ket{0}^{\otimes n}=\sqrt{1-a}\ket{\psi_0}\ket{0}+\sqrt{a}\ket{\psi_1}\ket{1},
\end{equation}
where $\ket{\psi_0}$ and $\ket{\psi_1}$ are unknown $(n-1)$-qubit quantum states. 
The goal is to estimate $a$ by measuring the output state of single or multiple quantum circuits that contain $U_a$ and $U_a^\dagger$. 
The performance of the QAE algorithm is evaluated by the relationship between the root mean squared estimation error (RMSE) $\varepsilon$ and the total number $N$ of queries to $U_a$ and $U_a^\dagger$. 
Notably, the conventional QAE algorithms~\cite{suzuki2020amplitude, grinko2021iterative, aaronson2020quantum,nakaji2020faster,intallura2023survey,labib2024quantum} achieve the Heisenberg-limited (HL) scaling $N = \mathcal{O}(1/\varepsilon)$ or the near-HL one $N = \tilde{\mathcal{O}}(1/\varepsilon)$ (where $\tilde{\mathcal{O}}$ suppresses logarithmic factors), over the classical scaling $\mathcal{O}(1/\varepsilon^2)$.
However, those QAE algorithms require applying $U_a$ and $U_a^\dagger$ sequentially on a single circuit; the total number of sequential queries of $U_a$ and $U_a^\dagger$ on a single circuit, which we call the depth, scales as $\mathcal{O}(1/\varepsilon)$, and this makes those QAE challenging to implement.

\begin{figure}[H]
    \centering
    \includegraphics[width=\columnwidth]{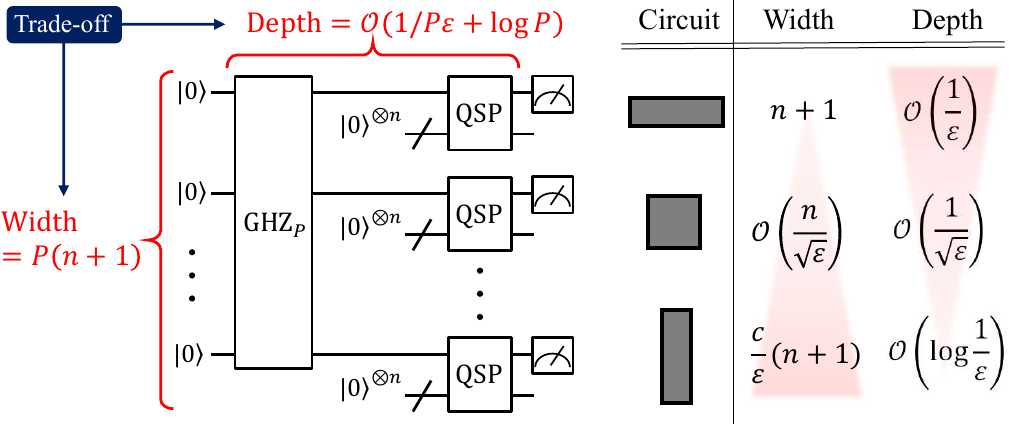}
    \caption{
        Quantum circuit of PAE, where $P \in [1, c/\varepsilon]$ represents the factor of parallelization with $c$ a constant. ``Width'' denotes the total number of qubits. $P$ is tuned to control the trade-off between total qubits and depth, as shown in several cases; Theorem~\ref{thm_PAE} in Introduction states the extreme log-depth case with $P=\lceil 1/\varepsilon \rceil$. 
        The ${\rm GHZ}_P$ operator prepares a $P$-qubit GHZ state, $(\ket{0}^{\otimes P}+\ket{1}^{\otimes P})/\sqrt{2}$.
        The QSP operator denotes an engineered phase shifter constructed by quantum signal processing (QSP), represented as $V_{\varphi, T}$ in the main text.
    }
    \label{fig_circuit_overview}
\end{figure}

Reducing circuit depth---even at the expense of additional qubits---is an effective approach for enhancing the implementability of quantum algorithms, which is thus a central paradigm in quantum algorithm synthesis~\cite{cleve2000fast, moore2001parallel, hoyer2005quantum, pham20132d, jiang2020optimal, zhang2024parallel, martyn2024parallel, wang2024tight, Quek2024multivariatetrace, cui2025unitary, mcardle2025fast}.
However, for QAE, few works pursue this direction~\cite{giurgica2022low,Rall2023amplitudeestimation,Vu2025low}. 
Refs.~\cite{giurgica2022low, Vu2025low} provide an example that achieves a depth of {${\mathcal{O}}(1/\varepsilon^{1-\kappa})$} with some constant $\kappa$, but it requires the total queries of $\tilde{\mathcal{O}}(1/\varepsilon^{1+\kappa})$, which is strictly bigger than the near HL scaling.
Overall, there has been no QAE algorithm achieving $N=\tilde{\mathcal{O}}(1/\varepsilon)$ for any $a \in [0, 1]$ with the use of quantum circuits whose maximal depth is sublinear in $1/\varepsilon$. 
In particular, there has been no log-depth QAE algorithm that achieves $N=\tilde{\mathcal{O}}(1/\varepsilon)$.

Intuitively, applying the quantum metrological parallel strategy to the QAE setting might work to solve the above-mentioned problems, because the Grover operator $Q$ in the QAE algorithms is a rotation gate with the angle $2\arcsin\sqrt{a}$. 
However, there is a tough obstacle; in the QAE problem, the eigenstates of $Q$ are not generally accessible unlike the conventional metrology setting, implying that the phase kick-back (from the system to the probe) technique cannot be directly applied.

In this paper, we leverage QSP~\cite{low2017optimal} in a specific way to overcome this issue, thereby presenting a new QAE algorithm---parallel amplitude estimation (PAE)---that achieves the desirable scaling in both the queries and the circuit depth; the following theorem is a special case achieving the log-depth circuit.

\begin{theorem}[Parallel amplitude estimation; log-depth case]
    Let $\varepsilon\in (0,1)$. 
    There exists a quantum algorithm that estimates $a \in [0, 1]$ encoded in $U_a$ 
    within the RMSE $\varepsilon$, using $N=\mathcal{O}(\varepsilon^{-1}\log(1/\varepsilon))$ queries to $U_a$ and $U_a^\dagger$ in total and $\lceil 1/\varepsilon \rceil(n+1)$-qubit quantum circuits with maximum circuit depth of $\mathcal{O}(\log(1/\varepsilon))$.
\end{theorem}

\noindent
That is, PAE resolves the above-mentioned open problem;
PAE can achieve the near HL scaling, $N=\tilde{\mathcal{O}}(1/\varepsilon)$, using quantum circuits with exponentially shallow depth $\mathcal{O}(\log (1/\varepsilon))$.
A resource comparison with conventional QAE~\cite{brassard2002quantum, suzuki2020amplitude, grinko2021iterative, giurgica2022low} is given in Supplemental Material (SM) Sec.~\ref{sec_comparison_QAE}.

Theorem~\ref{thm_PAE} can be generalized (the statement will be shown later), and Fig.~\ref{fig_circuit_overview} depicts the circuit of that general PAE algorithm. 
We can freely choose the parallelization factor $P$ in $[1, \mathcal{O}(1/\varepsilon)]$, and the resulting depth becomes $\mathcal{O}(1/(P\varepsilon)+\log P)$.
This depth scaling seems to be inconsistent with the previous lower bound $\Omega(1/(\varepsilon \sqrt{P}))$ in a parallel approximate counting problem~\cite{burchard2019lower}, which can be solved by PAE. 
However, we point out that the original derivation of this previous bound is incorrect. 
We then derive the corrected lower bound of query depth with $1/P$ dependence in a parallel approximate counting problem via the parallel quantum adversary method; see Theorem~\ref{thm:optimality_PAE} or a more general result in Appendix~\ref{sec_parallel_query}. 
As a result, we prove that the PAE algorithm can solve this problem and essentially matches the corrected lower bound. 
We also mention the consistency of PAE with the impossibility of efficient parallelization in quantum search at Appendix~\ref{sec_parallel_query}.

The notable parallel structure in Fig.~\ref{fig_circuit_overview} indeed comes from the parallel strategy in quantum metrology.
This represents an important feature of PAE; a (large) entanglement between multiple systems is needed only at the beginning of the circuit for preparing the $P$-qubit GHZ state $\ket{{\rm GHZ}_P}=(\ket{0}^{\otimes P}+\ket{1}^{\otimes P})/\sqrt{2}$.
After this, the circuit has a completely separable structure including the final measurement. 
This indicates that our method can be executed in parallel using multiple $\mathcal{O}(n)$-qubit quantum computers with the pre-shared entangled state $\ket{{\rm GHZ}_P}$, which can be generated with a logarithmic depth in $P$~\cite{cruz2019efficient, mooney2021generation}.
For this reason, our method is suitable for device implementation, especially in a form of distributed quantum computing~\cite{cirac1999distributed, yimsiriwattana2004generalized, van2006distributed, beals2013efficient, caleffi2024distributed, barral2025review}.

\textit{Parallel strategy in quantum metrology.}---
The standard problem addressed by the parallel strategy \cite{giovannetti2006quantum} is the estimation of an unknown phase $\varphi$ embedded in a unitary operator $U_{\varphi} := e^{i \varphi H}$. 
The crucial assumption is that the corresponding eigenstate of Hamiltonian $H$ can be prepared, i.e., $U_{\varphi}\ket{\varphi} = e^{i\varphi}\ket{\varphi}$. 
A canonical procedure of the parallel strategy is that we first prepare $\ket{{\rm GHZ}_P}$ together with $\ket{\varphi}^{\otimes P}$ and then apply the controlled-unitary ${\rm c}U_{\varphi}=\ket{0}\bra{0}\otimes \bm{1} + \ket{1}\bra{1}\otimes U_\varphi$ in parallel:
\begin{align}
    \label{eq_signal_multiplication}
    {\rm c}U_{\varphi}^{\otimes P} \ket{{\rm GHZ}_P} \ket{\varphi}^{\otimes P}
       = \frac{\ket{0}^{\otimes P}+e^{iP\varphi}\ket{1}^{\otimes P}}{\sqrt{2}}
              \ket{\varphi}^{\otimes P}.
\end{align}
Thus, the phase $\varphi$ is effectively kick-backed with multiplicative factor $P$, enabling the quantum-enhanced estimation of $\varphi$ to achieve the HL scaling in $P$ \cite{giovannetti2006quantum}. 
Note again that the above operation is doable if $\ket{\varphi}$ is available, while, if not, the possibility of doing a similar phase kick-back technique is non-trivial.
This is the main reason why the direct application of the parallel strategy to the QAE problem is a significant challenge. 
Below we describe this fact in detail.

\textit{Challenges of the parallel strategy for amplitude estimation.}---
In the QAE problem, the following Grover operator has an important role:
\begin{align}
    \label{eq_def_Q}
    Q := U_0 U_a^{\dagger} U_f U_a,
\end{align}
where $U_0 := 2\ket{0}^{\otimes n}\bra{0}^{\otimes n}-\bm{1}^{\otimes n}$, $U_f := 2 \times \bm{1}^{\otimes n-1} \otimes \ket{0} \bra{0}-\bm{1}^{\otimes n}$, and $\bm{1}$ is an identity operator.
$Q$ acts as $Q\ket{0}^{\otimes n} = \cos{2\theta}\ket{0}^{\otimes n} + \sin{2\theta}\ket{\psi}$, where $\theta := \arcsin{\sqrt{a}}$ and $\ket{\psi}$ is a quantum state orthogonal to $\ket{0}^{\otimes n}$~\cite{low2019hamiltonian,uno2021modified}. 
In the subspace spanned by $\ket{0}^{\otimes n}$ and $\ket{\psi}$, called ``Grover plane'', $Q$ functions as a rotation $e^{-i2\theta \overline{Y}}$ for the Pauli $\overline{Y}$ defined in this subspace. 
$Q$ has the eigenstates $\ket{Q_\pm} := (\ket{0}^{\otimes n} \pm i \ket{\psi})/\sqrt{2}$, which satisfy
\begin{align}
    \label{eq_operation_Q}
    Q\ket{Q_\pm} =  e^{\mp 2i\theta}\ket{Q_\pm}.
\end{align}
To realize the parallel strategy in QAE, we consider the controlled Grover operator ${\rm c}Q := \ket{0}\bra{0}_b \otimes \bm{1}_s + \ket{1}\bra{1}_b \otimes Q$, where $b$ and $s$ are indices corresponding to the ancilla qubit and the $n$-qubit system.
If the input state $\ket{{\rm GHZ}_P}_b \otimes \ket{Q_\sigma}_s^{\otimes P}$ ($\sigma=\pm$) can be prepared, applying ${\rm c}Q^{\otimes P}$ results in a signal multiplication similar to Eq.~\eqref{eq_signal_multiplication}.
However, in QAE, only the black-box operation $Q$ (or $U_a$ and $U_a^\dagger$) is given, and the eigenstates $\ket{Q_\pm}$ are generally unknown, meaning that the phase kick-back technique with ${\rm c}Q$ cannot be directly applied.
There are two previous approaches for addressing this issue: preparing $\ket{Q_\pm}$ assuming a sufficiently large amplitude~\cite{knill2007optimal}, or generating a particularly structured $\ket{Q_\pm}$~\cite{braun2022error}.
Unlike these approaches, we design a general and efficient parallel estimation method for arbitrary $a \in [0,1]$ and black boxes $U_{a},U_{a}^\dagger$.

\textit{Modified phase shifter via QSP.}---
To avoid preparing unknown states, we convert ${\rm c}Q$ into an engineered phase shifter which encodes the target parameter $a$ into the relative phase between {\it known} eigenstates.
The key idea of our approach is to make the eigenphases of $Q$ degenerate in the Grover plane, a technique that has also been employed in other contexts~\cite{low2017optimal, low2019hamiltonian}. 
Now, ${\rm c}Q$ can be expressed as 
\begin{align}
\label{eq_cQ expression}
    {\rm c}Q =  \sum_{\sigma}e^{-i\sigma\theta}
                            \begin{pmatrix}
                                e^{i\sigma\theta}&0\\
                                0&e^{-i\sigma\theta}
                            \end{pmatrix}_b
                            \otimes \ket{Q_{\sigma}}\bra{Q_\sigma}_s,
\end{align}
where $\sigma\in\{+,-\}$ and we omit terms acting outside the Grover plane.
Suppose we have an operation to transform the eigenphases $\sigma \theta$ to 
$h(\sigma \theta) = -T\cos(2\sigma \theta)$ and remove the global phase; then $cQ$ becomes
\begin{align*}
    &\sum_{\sigma}
        \begin{pmatrix}
            e^{-iT\cos(2\sigma\theta)}&0\\
            0&e^{iT\cos(2\sigma\theta)}
        \end{pmatrix}_b
        \otimes \ket{Q_{\sigma}}\bra{Q_\sigma}_s \nonumber \\
    & = 
        \begin{pmatrix}
            e^{-iT\varphi/2}&0\\
            0&e^{iT\varphi/2}
        \end{pmatrix}_b
        \otimes \sum_{\sigma} \ket{Q_{\sigma}}\bra{Q_\sigma}_s,
\end{align*}
where $T$ represents a tunable time duration and $\varphi := 2\cos{2\theta} = 2(1-2a)$. 
Hence, this procedure results in the following transformation: 
\begin{equation}
    \label{eq_def_Vphi}
    {\rm c}Q\mapsto \widetilde{V}_{\varphi, T}:=\begin{pmatrix}
        e^{-iT\varphi/2}&0\\
        0&e^{iT\varphi/2}
    \end{pmatrix}_b \otimes \overline{\bm{1}}_s,
\end{equation}
where $\overline{\bm{1}}_s$ is the identity operator on the Grover plane, and terms acting outside the Grover plane are omitted. 
Note that $\overline{\bm{1}}_s = \sum_{\sigma \in \{ +, - \}} \ket{Q_\sigma}\bra{Q_\sigma}_s$, because $\ket{Q_+}$ and $\ket{Q_-}$ form an orthogonal basis on the Grover plane.
Consequently, after preparing an arbitrary state, particularly a \textit{known} state such as $\ket{0}_s^{\otimes n}$ on the Grover plane, 
$\widetilde{V}_{\varphi, T}$ acts as the relative phase shifter of $T\varphi$ for any state in $b$.

The procedure described above can be approximately realized using QSP~\cite{low2017optimal, low2017quantum, martyn2021grand}, which is a general method for performing polynomial transformations on operator eigenvalues.
In our setting, the target operator is ${\rm c}Q$, and we focus only on the eigenvalues of $\ket{Q_\pm}$, whereas QSP transforms all eigenvalues.
Moreover, while standard applications of QSP require post-selection, our construction does not involve any post-selection.
Appendix~\ref{sec_operator_transformation} outlines the construction of the approximating unitary $V_{\varphi, T}$; details are in SM Sec.~\ref{sec_operator_transformation_detail}.
To quantify the resource requirements for this transformation, we introduce the following Lemma~\ref{lm_operator_transformation}:

\begin{lemma}
[Query complexity for constructing $V_{\varphi, T}$]
\label{lm_operator_transformation}
For any oracle conversion error $\varepsilon_{\rm oc}\in (0, 1)$ and 
any $j \in \{0, 1\}$,
there exists a quantum algorithm that constructs an operator $V_{\varphi, T}$ such that 
\[
\Big\lVert \left( V_{\varphi, T} - \widetilde{V}_{\varphi, T} \right) \ket{j}_b \ket{0}^{\otimes n}_s \Big\rVert < \varepsilon_{\rm oc},
\]
using ${\rm c}Q$ and ${\rm c}Q^{\dagger}$ a total of $L = \mathcal{O}(T + \log(1/\varepsilon_{\rm oc}))$ times.
\end{lemma}

Lemma~\ref{lm_operator_transformation} is derived by applying the theory of QSP~\cite{low2017optimal, gilyen2019optimizing, low2019hamiltonian} to this operator transformation (see SM Sec.~\ref{sec_proof_oc} for details).
Since $V_{\varphi, T}$ consists of a total of $L+2$ queries to $U_a$ and $U_a^\dagger$ (see Fig.~\ref{fig_circuit_Vph} in Appendix~\ref{sec_operator_transformation}), we can achieve an approximation error of $\varepsilon_{\rm oc}$ with a logarithmic number of (control-free) operations of $U_a$ and $U_a^\dagger$. 
As a result, this cost accounts for the additional $\log(1/\varepsilon)$ factor in the query complexity stated in Theorem~\ref{thm_PAE}.

With $V_{\varphi, T}$, we can perform a similar signal amplification to Eq.~\eqref{eq_signal_multiplication}:
\begin{align}
    \label{eq_signal_multiplication_with_Q}
    &\ket{\Psi(M=PT)} := V_{\varphi, T}^{\otimes P} \ket{{\rm GHZ}_P}_b\ket{0}^{\otimes n P}_s \nonumber \\
    &\hspace{30pt} \approx \frac{e^{-iM \varphi/2}\ket{0}^{\otimes P}_b + e^{iM \varphi/2} \ket{1}^{\otimes P}_b }{\sqrt{2}} \otimes \ket{0}^{\otimes n P}_s,
\end{align}
where the approximation comes from $V_{\varphi, T} \approx \widetilde{V}_{\varphi, T}$.
That is, the phase $\varphi$ is successfully kick backed to the ancilla space with multiplicative enhancement factor $M=PT$. 
Again, this is the transformation on the Grover plane. 
Also note that $\ket{0}^{\otimes n P}_s$ can be prepared without knowing $\ket{Q_{\pm}}.$
The quantum circuit for this operation is illustrated in Fig.~\ref{fig_circuit_overview}, where the QSP operator denotes $V_{\varphi, T}$.
Note that for a given $M$, the parameters $P$ and $T$ are chosen according to the available quantum resources.

\textit{Amplitude estimation with parallel strategy.}---
In PAE, we estimate $\varphi := 2(1 - 2a)$, approximately embedded in $V_{\varphi, T}$.
Specifically, to resolve phase ambiguity due to the periodicity in Eq.~\eqref{eq_signal_multiplication_with_Q}, we leverage the robust phase estimation (RPE) method~\cite{higgins2009demonstrating, kimmel2015robust} through the following quantum-enhanced measurement in the parallel strategy. 
The concrete procedure for estimating $\varphi$ using RPE is as follows. 
Let $K$ be some positive integer.
(i) For each $k \in \{1, 2, ..., K \}$, prepare $\ket{\Psi(M_k=2^{k-1})}$ with any pair ($P_k, T_k$) satisfying $P_kT_k=2^{k-1}$.
Then perform each of the two projective measurements including the bases
$\{ \ket{\pm_{P_k}}_b := (\ket{0}^{\otimes {P_k}}_b \pm \ket{1}^{\otimes {P_k}}_b)/\sqrt{2} \}$ or $\{ \ket{\pm i_{P_k}}_b := (\ket{0}^{\otimes {P_k}}_b \pm i\ket{1}^{\otimes {P_k}}_b)/\sqrt{2} \}$ on the ancilla subsystem $\nu_k$ times, and record the number of trials in which the outcomes are $\ket{+_{P_k}}_b$ and $\ket{+i_{P_k}}_b$, respectively.
(ii) Conduct classical postprocessing on the results of (i) to estimate the phase.
The pseudocode for the classical post-processing in (ii) is presented in Appendix~\ref{sec_pseudocode}, with details in SM Sec.~\ref{sec_RPE}.
This post-processing is very simple and its computational cost is almost negligible. 

Notably, the outcomes of the two projective measurements in (i) can be reproduced by measuring each ancilla qubit~\cite{belliardo2020achieving}.
The probability of obtaining an even number of 1s from X-measurements on the ancilla qubits of $\ket{\Psi(M_k)}$ equals the projection probability onto $\ket{+_{P_k}}_b$.
After applying $e^{i\pi Z/ 4}$ to the first ancilla qubit, the probability corresponds to that of finding $\ket{+i_{P_k}}_b$.
Therefore, in PAE, the only quantum operation across $P$ parallel systems is the preparation of $\ket{{\rm GHZ}_{P}}_b$.
Moreover, all $k$-th processes in (i) are independent and can run in parallel. 
The pseudocode of PAE is provided in Appendix~\ref{sec_pseudocode}.

Importantly, the RPE procedure works well even if the quantum state preparation and/or measurement contain some small errors. 
Here, we assume that the probabilities of obtaining the outcomes corresponding to projective measurements onto $\ket{+_{P_k}}_b$ and $\ket{+i_{P_k}}_b$ are given by $p_{+,k} := (1 + \cos{M_k \varphi})/2 + \beta_{+, k}$ and $p_{i,k} := (1 + \sin{M_k \varphi}) / 2 + \beta_{i, k}$, respectively, where $\beta_{r, k}$ (for $r \in \{ +, i \}$) denotes the bias in the measurement probability caused by the approximation error of $V_{\varphi,T}$ or the computational error. 
Due to the robustness of RPE, one can achieve the HL scaling for the estimation of $\varphi$ if $|\beta_{r, k}| < \sqrt{6}/8$~\cite{kimmel2015robust, belliardo2020achieving}.
Based on the discussion in Ref.~\cite{belliardo2020achieving}, we have the following lemma regarding $\beta_{r, k}$ and the mean squared estimation error (MSE) upper bound:

\begin{lemma}
[MSE upper bound of RPE~\cite{belliardo2020achieving}]
\label{lm_RPE_RMSE}
Suppose the measurement bias parameters $\{\beta_{r, k}\}$ satisfy $\sup_{r, k} \{|\beta_{r, k}|\} := \beta < \sqrt{6}/8$.
Then, the RPE procedure (i)--(ii) returns the phase estimate $\hat{\varphi} \in (-\pi, \pi]$ such that its mean squared error (MSE) satisfies
\begin{align}
    \label{eq_eps_upper-bound}
    \mathbb{E}\left[ (\hat{\varphi} - \varphi)^2 \right] \le \left( \cfrac{2\pi}{3} \right)^2 \left( \cfrac{1}{4^K} + \sum_{k=1}^K \cfrac{e^{-2 \nu_k (\sqrt{6}/8 - \beta)^2}}{4^{k-4}} \right).
\end{align}
\end{lemma}

We now provide Theorem 1 for the general case of $P$, followed by the proof sketch.

\setcounter{theorem}{0}
\begin{theorem}[Parallel amplitude estimation; general case]
    \label{thm_PAE}
    Let $\varepsilon\in (0,1)$, and let $P$ be any positive integer.
    There exists a quantum algorithm that estimates $a \in [0, 1]$ encoded in $U_a$ (Eq.~\eqref{eq:quantum_oracle_qae}) within the RMSE $\varepsilon$, using $N=\mathcal{O}\left(1/\varepsilon + P \log{P}\right)$ queries to $U_a$ and $U_a^\dagger$ in total.
    This quantum algorithm uses $P(n+1)$-qubit quantum circuits with the structure depicted in Fig.~\ref{fig_circuit_overview} and the maximum circuit depth of $\mathcal{O}(1/(\varepsilon P) + \log P)$.
\end{theorem}

\textit{Proof sketch of Theorem~\ref{thm_PAE}.}--- 
The goal is, in the framework of RPE, to compute the necessary resources (circuit depth and width) such that the right hand side of Eq.~\eqref{eq_eps_upper-bound} in Lemma~\ref{lm_RPE_RMSE} is at most $\varepsilon^2$. 
For any positive integer $P$, we consider the circuit such that $P_k\le P$ for all $k$.
When applying $P_k$ copies of $V_{\varphi,T_k}$ in parallel as in Fig.~\ref{fig_circuit_overview}, the trace distance between the ideal state and the approximate (implemented) state is $\mathcal{O}(P_k\varepsilon_{\rm oc})$ via a telescoping-sum argument.
Since the trace distance between two quantum states upper bounds the total variation distance for any POVM \cite{nielsen2010quantum}, and $\beta_{r,k}$ is defined as a (two-outcome) measurement-probability bias, we obtain the bound $|\beta_{r,k}| = \mathcal{O}(P_k\varepsilon_{\rm oc})$.
By Lemma~\ref{lm_operator_transformation}, choosing $L_k=\mathcal{O}(T_k+\log(P_k/\beta))$  guarantees $\varepsilon_{\rm oc}=\mathcal{O}(\beta/P_k)$, and thus $|\beta_{r,k}|<\beta<\sqrt{6}/8$.
Choosing $K=\mathcal{O}(\log(1/\varepsilon))$ and $\nu_k=\mathcal{O}(K-k)$, Lemma~\ref{lm_RPE_RMSE} yields $\mathbb{E}[(\hat{\varphi}-\varphi)^2]\le \varepsilon^2$, and since $\varphi := 2(1 - 2a)$, we have $\sqrt{\mathbb{E}[(\hat{a} - a)^2]} < \varepsilon$. 
The total query count is $N = 2\sum_{k=1}^{K} \nu_k (L_k+2) P_k$, where the prefactor $2$ accounts for the two measurement settings $r\in\{+,\;i\}$ in RPE.
Choosing $(T_k,P_k)$ appropriately under the constraint $M_k=P_kT_k=2^{k-1}$ yields $N=\mathcal{O}(1/\varepsilon + P\log P)$.
Since $V_{\varphi,T_k}$ requires $\mathcal{O}(L_k)$ sequential oracle calls, and $\ket{{\rm GHZ}_P}_b$ can be prepared with $\mathcal{O}(\log P)$-depth circuit~\cite{cruz2019efficient, mooney2021generation}, the overall depth becomes $\mathcal{O}(N/P)= \mathcal{O}(1/(\varepsilon P)+\log P)$. 
Since at most $P$ instances of an $(n + 1)$-qubit system are arranged in parallel, the total number of qubits is $P(n+1)$. 
The detailed proof is in SM Sec.~\ref{sec_proof_theorem1}.
\hfill\qedsymbol

Note that, PAE estimates $a$ directly rather than $\theta$, so its total query complexity lacks the $\sqrt{a(1-a)}$ factor that appears in conventional QAE~\cite{brassard2002quantum}; if $a\approx0$ or $a\approx 1$ is known in advance, conventional QAE can be more efficient.

\begin{figure*}[t]
    \includegraphics[width=\textwidth]{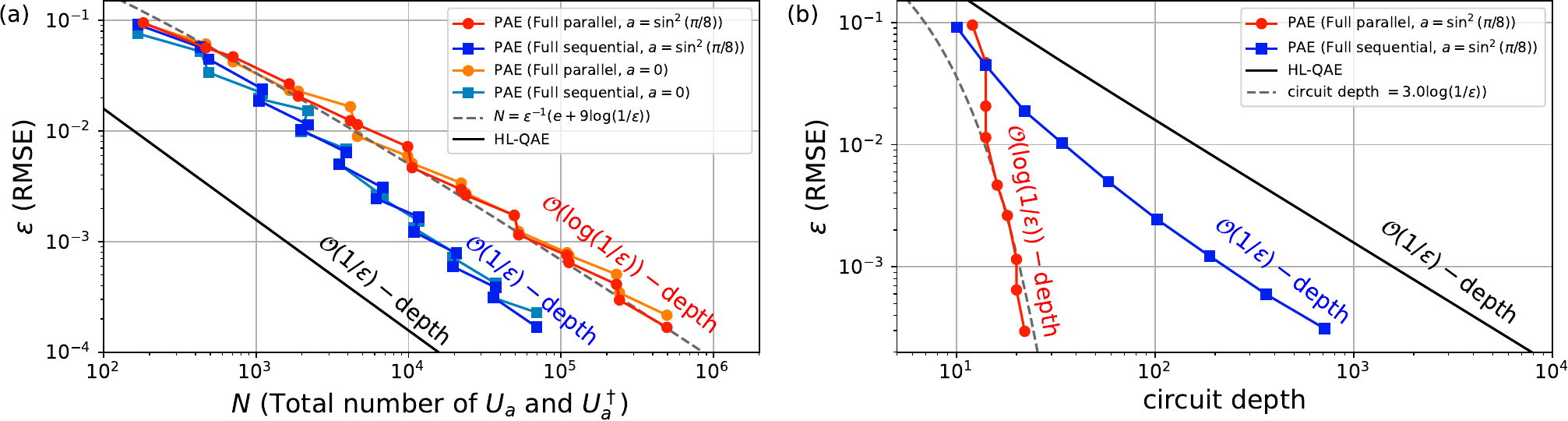}
    \caption{
        (a) Relationship between the number of queries to $U_a$ and $U_a^{\dagger}$ and the estimation error (RMSE), and (b) relationships between the circuit depth and the RMSE.
        In both graphs, the gray dashed line shows a simple fitting result for the ``Full parallel'' case with $a = \sin^2(\pi/8)$.
        The non-smooth behavior in (a) arises from the discrete choices of $K$ in the RPE procedure, as detailed in SM Sec.~\ref{sec_detail_numerical_experiment}.
    }
    \label{fig_graph_query_comp}
\end{figure*}

\textit{Optimality of PAE.}---
To see the optimality of PAE, we revisit an approximate counting problem.
The goal is to estimate the number $N_t$ of marked items in the size-$N_d$ database within  a relative error $\varepsilon_{\rm rel}$. 
In parallel approximate counting, Ref.~\cite{burchard2019lower} claims that the lower bound of $P$-parallel query complexity (the minimal depth of $P$-parallel queries, see the formal definition in Appendix~\ref{sec_parallel_query}) is $\varepsilon_{\rm rel}^{-1}\sqrt{N_d/(PN_t)}$ up to a constant factor.
However, we have identified an error in its derivation; after correcting this, we obtain the following theorem.
\begin{theorem}[Lower bound in parallel approximate counting]\label{thm:optimality_PAE}
    Let us consider an approximate counting problem for the number $N_t\in(\Theta(N_d),N_d/2]$ of marked items in a size-$N_d$ database within a relative error $\varepsilon_{\rm rel}\in (\Omega(N_d^{-1}),1/2)$.
    Then, for any quantum algorithm solving this problem with high probability, the $P$-parallel query complexity is $\Omega(\varepsilon_{\rm rel}^{-1}/P)$ for any $P\in [1,\Theta(N_d))$.
\end{theorem}
\noindent
We provide the specification of that proof error, the derivation of Theorem~\ref{thm:optimality_PAE}, and a more general lower bound in Appendix~\ref{sec_parallel_query}.
Importantly, the corrected lower bound indicates the $1/P$ scaling, as opposed to the previous argument.
Note now that the approximate counting problem in Theorem~\ref{thm:optimality_PAE} can be solved by amplitude estimation algorithms that estimate $N_t/N_d$ within the additive error $\varepsilon=\varepsilon_{\rm rel}\cdot\Theta(1)$. 
In particular, the PAE algorithm using the standard QAE oracle $U_a$ with the parameter $a=N_t/N_d$~\cite{nielsen2010quantum} solves this problem with high probability by making $\mathcal{O}(\varepsilon_{\rm rel}^{-1}/P+\log P)$ $P$-parallel queries.
Therefore, PAE is optimal (up to the additive log factor) among all algorithms (including general QAE algorithms) that can solve approximate counting in the nontrivial parameter regime specified in Theorem~\ref{thm:optimality_PAE}.

\textit{Numerical experiment.}---
We numerically evaluate PAE’s total query count and circuit depth.
The computational details are presented in SM Sec.~\ref{sec_detail_numerical_experiment}.
The simulation code is available on GitHub~\cite{oshio2025Near_github}.
We consider two choices of $P_k$ and $T_k$: 
(i)~\textit{Full parallel:} fix $T_k = 1~\forall k$ and set $P_k=2^{k-1}$, and 
(ii)~\textit{Full sequential:} fix $P_k = 1~\forall k$ and set $T_k=2^{k-1}$. 
For comparison, we also plot the query counts of ``HL-QAE''~\cite{koizumi2025comprehensive} ($\varepsilon = \pi/2(N-1)$), which is the most query-efficient QAE proposed to date.

Figure~\ref{fig_graph_query_comp}(a) shows the query counts versus RMSE. 
In the full sequential case (ii), PAE achieves the HL scaling $N = \mathcal{O}(1/\varepsilon)$. 
In the full parallel case (i), the scaling remains HL with logarithmic overhead, 
$N=\mathcal{O}(\varepsilon^{-1}\log (1/\varepsilon))$, consistent with Theorem~\ref{thm_PAE}. 
This overhead leads to about 4 times bigger queries $N$ for $\varepsilon = 10^{-3}$, but we recall that the full parallel PAE works only with log-depth circuit. 
This is clearly seen in Fig.~\ref{fig_graph_query_comp}(b) showing the circuit depth versus RMSE. 
Actually, in the case (i), the depth scales logarithmically in $1/\varepsilon$, also in agreement with Theorem~\ref{thm_PAE}. 
In contrast, the PAE with the case (ii) needs $\mathcal{O}(1/\varepsilon)$ depth, which is the same as HL-QAE. 
Note however that, compared to HL-QAE which requires $\mathcal{O}(\log (1/\varepsilon))$ ancilla qubits, this PAE achieves roughly $1/6$ circuit depth for $\varepsilon \lesssim 5\times10^{-3}$ while using only a single ancilla qubit, at the cost of about 10 times increase in $N$ as observed in Fig.~\ref{fig_graph_query_comp}(a).

\textit{Summary and discussion.}---
PAE’s key feature is its capability of controlling the trade-off between circuit depth and qubit count. 
This may enable pursuing HL scaling for amplitude estimation even on depth-limited early fault-tolerant quantum computing devices.  
For instance, for the case $\varepsilon = 10^{-3}$, PAE with $P=64$ needs quantum circuits of depth 20 assisted by a 64-qubit GHZ state.
In addition, under the assumption that the wall-clock time of a quantum algorithm is determined by the depth of its quantum circuit, leveraging PAE to increase parallelism allows for a reduction in total computation time compared to conventional (non-parallel) methods. 
Since amplitude estimation can be seen as a metrological estimation task, it is natural from the viewpoint of quantum metrology to achieve the $1/P$ scaling for the parallelization $P$. 
Further exploring quantum algorithms that admit $1/P$ scaling is an important future direction, while many parallel quantum algorithms fail to achieve this scaling~\cite{zalka1999grover, grover2004quantum, jeffery2017optimal} and our technique seems specific to amplitude estimation.

\begin{acknowledgments}
We thank Alexander Dalzell and Ronald de Wolf for helpful comments. 
We also thank the members of the Keio University Quantum Computing Center for many helpful discussions and feedback.
This work is supported by MEXT Quantum Leap Flagship Program Grant No. JPMXS0118067285 and JPMXS0120319794. 
K.O. acknowledges support by SIP Grant Number JPJ012367.
K.W. was supported by JSPS KAKENHI Grant Number JP24KJ1963.
\end{acknowledgments}

\clearpage
\section*{Appendix}
\setcounter{section}{0}
\renewcommand{\thesection}{\Alph{section}}
\renewcommand{\thesubsection}{\thesection\arabic{subsection}}
\makeatletter
\renewcommand{\p@subsection}{}
\makeatother
\setcounter{equation}{0}
\renewcommand{\theequation}{\Alph{section}\arabic{equation}}

\section{Construction of $V_{\varphi, T}$ with QSP}
\label{sec_operator_transformation}

Using QSP, $\widetilde{V}_{\varphi, T}$ defined in Eq.~\eqref{eq_def_Vphi} can be approximated as $V_{\varphi, T}$ of the form:
\begin{align}
    \label{eq_def_Vph_detail}
    V_{\varphi, T} &= \prod^{L/2}_{l = 1} \left( R_x(\xi_{2l-1}') \otimes \bm{1}_s \right) W_Q^{\dagger} \left( R_x(-\xi_{2l-1}') \otimes \bm{1}_s \right)\nonumber \\ 
    &\hspace{30pt} \times \left( R_x(\xi_{2l}) \otimes \bm{1}_s \right) W_Q \left( R_x(-\xi_{2l}) \otimes \bm{1}_s \right), 
\end{align}
where $W_Q = {\rm c}Q \times R_z(\pi/2)$, $R_z(\xi) = e^{-i\xi Z_b / 2}$ and $R_x(\xi) = e^{-i\xi X_b / 2}$, with $Z_b$ and $X_b$ being the Pauli operators acting on the ancilla qubit.
$\vec{\xi}$ is a QSP hyperparameter, referred to as the angle sequence, chosen to ensure that $V_{\varphi, T} \approx \widetilde{V}_{\varphi, T}$.
Here, $\xi' = \xi + \pi$.
The circuit structure of $V_{\varphi, T}$ is illustrated in Fig.~\ref{fig_circuit_Vph}.
The detail of this construction is presented in SM Sec.~\ref{sec_operator_transformation_detail}.

\begin{figure}[H]
    \centering
    \includegraphics[width=0.95\columnwidth]{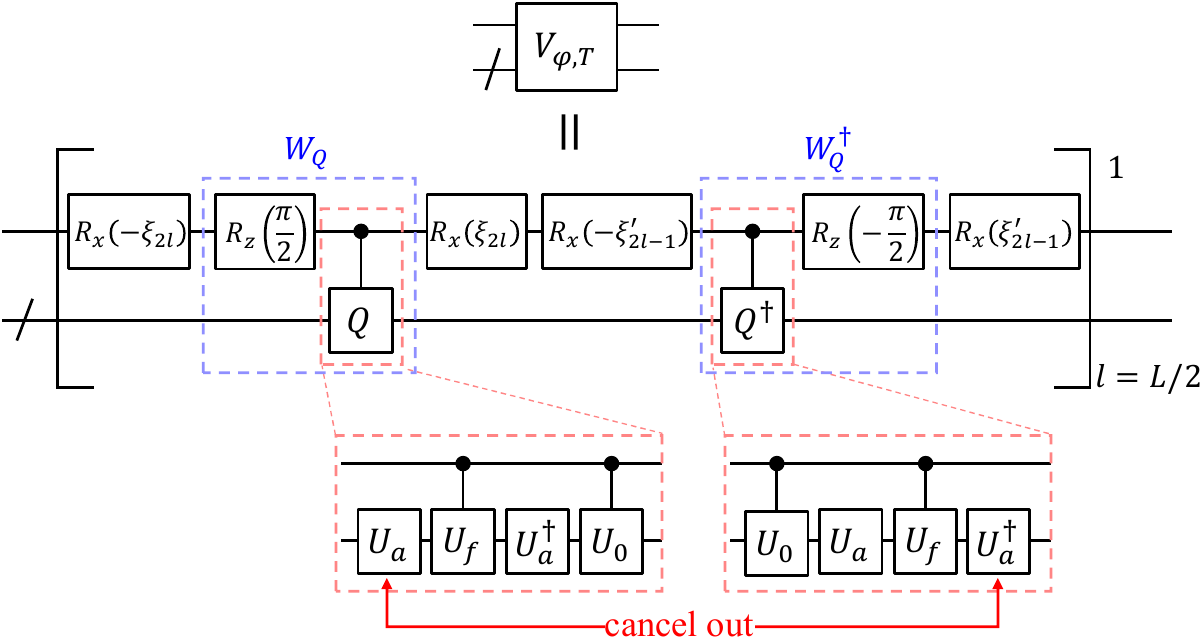}
    \caption{ 
        Construction of $V_{\varphi, T}$. 
        Here, $L$ is a function of $T$, and $\varphi$ is connected with the eigenphase $2\theta$ of $Q$ (in the Grover plane) by $\varphi=2\cos(2\theta)$.
    }
    \label{fig_circuit_Vph}
\end{figure}

Note that $U_a$ in $W_Q$ cancels out with the adjacent $U_a^{\dagger}$ in $W_Q^{\dagger}$.
Therefore, $V_{\varphi, T}$ contains a total of $L + 2$ applications of $U_a$ and $U_a^{\dagger}$.
To construct $V_{\varphi, T}$, it is also possible to employ the generalized QSP (GQSP)~\cite{motlagh2024generalized} instead of standard QSP.
While GQSP has been shown to halve the cost of Hamiltonian simulation~\cite{berry2024doubling}, it does not provide the same reduction in our setting, as the cancellation structure between $U_a$ and $U_a^\dagger$ does not arise when using GQSP.

\section{Pseudocodes}
\label{sec_pseudocode}
The complete PAE procedure and the classical post-processing in RPE are presented in Algorithms~\ref{alg_PAE} and~\ref{alg_RPE}, respectively.

In PAE, $P_k$ and $T_k$ can be chosen under the constraint $M_k = P_k T_k = 2^{k-1}$, depending on available quantum resources.
However, large $T_k$ may destabilize the computation of the QSP hyperparameters~\cite{haah2019product,chao2020finding}.
To address this issue, one can achieve the same phase amplification effect by applying $V_{\varphi,T}$ sequentially $S$ times, at the cost of replacing $\varepsilon_{\rm oc} \mapsto S\varepsilon_{\rm oc}$ in the error bound stated in Lemma~\ref{lm_operator_transformation} (see SM Sec.~\ref{sec_proof_oc} for details).
Due to this error bound modification, the query complexity becomes $ N = \mathcal{O}\left(\varepsilon^{-1} + PS \log{PS}\right)$, and the depth becomes $\mathcal{O}(1/(\varepsilon P) + S\log PS)$.

\begin{algorithm}[H]
    \caption{Parallel amplitude estimation}
    \label{alg_PAE}
    \begin{algorithmic}[1]
    \Require Operator $U_a, U_a^{\dagger}$, target RMSE $ \varepsilon \in (0, 1)$, target bias threshold $\beta \in (0, \sqrt{6}/8)$.
    \Ensure Estimate $\widehat{a}$.
    \State $K \gets \lceil \log_2(1/\varepsilon) \rceil + 6$ 
    \State Construct $Q$ defined in Eq.~\eqref{eq_def_Q} from $U_a$ and $U_a^{\dagger}$.
    \State Set $[P_1, P_2,..., P_K]$ and $[T_1, T_2,..., T_K]$ so that $P_kT_k = 2^{k-1}, \forall k \in \{1,2, \cdots, K\}$.
    \For{$k = 1, 2,..., K$} \ \textbackslash\textbackslash This for-loop can be parallelized
        \State $\nu_k \gets 1 + \lceil \log{6} \times (K - k)/2(\sqrt{6}/8-\beta)^2 \rceil$
        \State Construct $V_{\varphi, T_k}$ using $U_a$ and $U_a^\dagger$ for \par\hspace{-7pt} $L_k = \mathcal{O}(T_k + \log(P_k/\beta))$ times. 
        \State Perform $V_{\varphi, T_k}^{\otimes P_k}$ on initial state $\ket{{\rm GHZ}_{P_k}}_b\ket{0}^{\otimes n P_k}_s$, \par\hspace{-7pt} in the same manner as Fig.~\ref{fig_circuit_overview}.
        \State Perform two measurements ($\nu_{k}$ repetitions each):
        \par\hspace{-7pt} (i) X-measurement on each ancilla qubit;
        \par\hspace{-7pt} (ii) X-measurement on each ancilla qubit after
        \par\hspace{9pt}      applying $e^{i\pi Z/4}$ to the first ancilla qubit.
        \State Set $h_{+,k}$ and $h_{i,k}$ to the counts of even‑parity
        \par\hspace{-7pt} outcomes in cases (i) and (ii), respectively.
        \State Calculate $f_{+,k} = h_{+,k}/\nu_{k}$ and $f_{i,k} = h_{i,k}/\nu_{k}$.
    \EndFor 
    \State Obtain $\widehat{a}$ using Algorithm~\ref{alg_RPE} with $\{f_{+,k}\}_{k=1}^K$ and $\{f_{i,k}\}_{k=1}^K$.
    \end{algorithmic}
\end{algorithm}

\begin{algorithm}[H]
    \caption{Robust phase estimation (classical post-processing part)}
    \label{alg_RPE}
    \begin{algorithmic}[1]
    \Require Max. number of steps $K$, Observed probabilities $\{f_{+,k}\}_{k=1}^K$, $\{f_{i,k}\}_{k=1}^K$
    \Ensure Estimate $\widehat{\varphi} \in [-\pi, \pi) $
    \For{$k = 1, 2,..., K$}
        \State $\widehat{M_k\varphi_{k}'} \gets {\rm atan2}(2f_{i, k} - 1, 2f_{+, k} - 1) \in [0, 2\pi)$
        \State $\widehat{\varphi}'_{k, 0} = \widehat{M_k\varphi_{k}'} / M_k  \in [0, 2\pi/M_k)$.
        \If {$k = 1$}
            \State $\widehat{\varphi}'_1 \gets \widehat{\varphi}'_{1, 0}$
        \Else
            \State $\eta \gets \left\lfloor \cfrac{\widehat{\varphi}'_{k-1}}{\pi/2^{k-2}} \right\rfloor$
            \If{$\widehat{\varphi}'_{k-1} - \left( \widehat{\varphi}'_{k, 0} + (\eta - 1)\pi / 2^{k-2} \right) \le \pi/2^{k-1}$}
                \State $\widehat{\varphi}'_k \gets \widehat{\varphi}'_{k, 0} + (\eta - 1)\pi / 2^{k-2}$
            \ElsIf{$\left( \widehat{\varphi}'_{k, 0} + (\eta + 1)\pi / 2^{k-2} \right) - \widehat{\varphi}'_{k-1} < \pi/2^{k-1}$}
                \State $\widehat{\varphi}'_k \gets \widehat{\varphi}'_{k, 0} + (\eta + 1)\pi / 2^{k-2}$
            \Else
                \State $\widehat{\varphi}'_k \gets \widehat{\varphi}'_{k, 0} + \eta\pi / 2^{k-2}$
            \EndIf
        \EndIf
    \State $\widehat{\varphi}_k \gets \widehat{\varphi}'_k - 2\pi \left\lfloor \cfrac{\widehat{\varphi}'_k + \pi}{2\pi} \right\rfloor$
    \EndFor
    \State $\widehat{\varphi} \gets \widehat{\varphi}_{K}$    
    \end{algorithmic}
\end{algorithm}

\section{Consistency with existing parallel query lower bounds}
\label{sec_parallel_query}
Prior works~\cite{zalka1999grover, grover2004quantum, jeffery2017optimal, burchard2019lower} derived lower bounds on the $P$-parallel query complexity of unstructured quantum search and approximate counting.
A $P$-parallel query model allows each query step to consist of $P$ oracle queries in parallel.
Then, the (bounded-error) $P$-parallel query complexity of a function $f$ is defined as the minimal number (or depth) of such $P$-parallel queries needed among all quantum algorithms that output $f(x)$ with high probability for every input $x$ in a domain. 
When $P=1$, this definition reduces to the standard (sequential) query complexity.
In the query model, algorithms have access to an oracle that indicates whether a queried item is marked.
For example, the standard oracle is given by $O_x:\ket{j,b}\mapsto \ket{j,b\oplus x_j}$, where $b\in \{0,1\}$ and $x=x_1x_2\cdots$ is an input bit string~\cite{de2019quantum}.
Here, we compare the existing lower bounds of the $P$-parallel queries (depth) with that of the proposed PAE algorithm with error $\varepsilon$, denoted by 
\begin{equation}
    {\rm PAE}^{P}(\varepsilon)=\mathcal{O}\left(\frac{1}{\varepsilon P}+\log P\right).
\end{equation}

\subsection{Parallel quantum search}
\label{sec_parallel_query_search}
We first verify consistency with the lower bound of parallel quantum search.
For an unstructured search problem with a single marked item, any quantum algorithm with $P$-parallel queries has depth $\Omega(\sqrt{N_d/P})$, where $N_d$ is the database size~\cite{zalka1999grover, grover2004quantum}.
This can be rephrased as follows; the bounded-error $P$-parallel query complexity to compute the $N_d$-bit OR function is $\Omega(\sqrt{N_d/P})$~\cite{jeffery2017optimal}. 
Now, using the standard QAE oracle construction with the parameter $a=N_t/N_d$~\cite{nielsen2010quantum}, the PAE algorithm can estimate the number $N_t$ of marked items.
Thus, computing OR reduces to distinguishing $a=0$ from $a\geq 1/N_d$ in PAE with error $\varepsilon=\Theta(1/N_d)$.
As a result, the PAE algorithm requires $\mathcal{O}(N_d/P+\log P)$ $P$-parallel queries for OR.
This exceeds $\Omega(\sqrt{N_d/P})$ for $P< N_d$,
so there is no contradiction with the parallel search lower bound.

\subsection{Parallel approximate counting}
\label{sec_parallel_query_counting}
We next revisit the lower bound of parallel approximate counting in Ref.~\cite{burchard2019lower} and point out an error in its derivation.
The goal of approximate counting is to estimate the number of marked items $N_t$ among $N_d$ data items within a target relative error $\varepsilon_{\rm rel}$.
Ref.~\cite{burchard2019lower} claims that, for the relative error $\varepsilon_{\rm rel}$ and the (non-zero) number $N_t$ satisfying $N_t+\lceil\varepsilon_{\rm rel}N_t\rceil\leq N_d$~\footnote{In Ref.~\cite{burchard2019lower}, although the author considers the condition $N_t+\varepsilon_{\rm rel}N_t\leq N_d$, we here introduce the ceil function for properly defining two distinct sets $X$ and $Y$.}, any quantum algorithm needs $\Omega(\varepsilon_{\rm rel}^{-1}\sqrt{N_d/(PN_t)})$ $P$-parallel queries.
At first sight, this $1/\sqrt{P}$-dependence seems inconsistent with PAE when $N_d<PN_t$.
This is because by taking $(N_t/N_d)\varepsilon_{\rm rel}$ as $\varepsilon$ (despite $N_t$ being unknown a priori), PAE yields an estimate within error $\varepsilon_{\rm rel}$ by making $\mbox{PAE}^P((N_t/N_d)\varepsilon_{\rm rel})$, namely $1/P$-dependent, $P$-parallel queries.
However, this comparison does not make sense because we have found an error in the derivation of Ref.~\cite{burchard2019lower}.

The error is in the proof of Theorem 3 in Ref.~\cite{burchard2019lower}; the parameter $\ell$ in this paper, defined as the maximum size of sets for an (extended) quantum adversary method, is underestimated, which results in an overly strong lower bound.
Specifically, we find that the correct value of $\ell$ is
\begin{equation}
    \ell=\binom{N_d-N_t}{\lceil\varepsilon_{\rm rel}N_t\rceil}-\binom{N_d-N_t-P}{\lceil\varepsilon_{\rm rel}N_t\rceil}
\end{equation}
when $N_d-N_t-P\geq \lceil\varepsilon_{\rm rel}N_t\rceil$.
This is strictly larger than the previous evaluation $\ell=\binom{N_d-N_t-1}{\varepsilon_{\rm rel}N_t-1}$ even for $P=2$.
By using the correct value of $\ell$, we prove that in parallel approximate counting, the lower bound of $P$-parallel query complexity is $\Omega({\rm Q})$, where ${\rm Q}$ is defined as
\begin{align}
\label{Q in App}
        {\rm Q}&=\left[1-\frac{\dbinom{N_d - N_t - P}{\lceil\varepsilon_{\rm rel}N_t\rceil}}{\dbinom{N_d - N_t}{\lceil\varepsilon_{\rm rel}N_t\rceil}}\right]^{-\frac{1}{2}}\left[1-\frac{\dbinom{N^{\varepsilon_{\rm rel}}_t - P}{\lceil\varepsilon_{\rm rel}N_t\rceil}}{\dbinom{N^{\varepsilon_{\rm rel}}_t}{\lceil\varepsilon_{\rm rel}N_t\rceil}}\right]^{-\frac{1}{2}},
\end{align}
where $N_t^{\varepsilon_{\rm rel}}:=N_t+\lceil\varepsilon_{\rm rel}N_t\rceil$.
The proof is given in SM Sec.~\ref{sec_parallel_adv_bound}, which also includes the correct derivation of $\ell$.

The lower bound $\Omega({\rm Q})$ for approximate counting implies validity and optimality of PAE.
We summarize some features of $\rm Q$ below; see SM Sec.~\ref{sec_parallel_adv_bound} for detailed discussions. 
First, as a sanity check, we can confirm that $\rm Q$ is always upper bounded by the depth scaling of PAE as 
\begin{equation}
        {\rm Q}= \mathcal{O}\left(\mbox{PAE}^{P}\left(\varepsilon_{\rm rel}{N_t/N_d}\right)\right).
\end{equation}
In particular, this highlights the validity of the $1/P$ scaling in PAE, as opposed to the previous overly strong bound.
Next, in a nontrivial regime $P\leq \min\{N_t,N_d-N_t^{\varepsilon_{\rm rel}}\}$ where $P$ is not too large, we can simplify $\rm Q$ in Eq.~\eqref{Q in App} and identify a clear lower bound 
\begin{equation}
\label{eq:apdx_advlb_nontrivial}
        {\rm Q}=\Omega\left(\frac{1}{P}\frac{N_t}{\lceil \varepsilon_{\rm rel}N_t\rceil}\sqrt{\frac{N_d-N_t(1+\varepsilon_{\rm rel})}{N_t}}\right).
\end{equation}
Again, we can confirm the $1/P$ scaling explicitly.
This lower bound immediately proves Theorem~\ref{thm:optimality_PAE} in the main text, which shows the optimality of PAE.

\textit{Proof of Theorem~\ref{thm:optimality_PAE}.}---
We assume that $N_t \in (\Theta(N_d),N_d/2]$ and $\varepsilon_{\rm rel}\in (\Omega(N_d^{-1}),1/2)$.
Then, we have the following evaluations:
\begin{equation}
    \frac{N_d-N_t(1+\varepsilon_{\rm rel})}{N_t}=\Theta(1),~~~\frac{N_t}{\lceil \varepsilon_{\rm rel}N_t\rceil}=\Theta(1/\varepsilon_{\rm rel}).
\end{equation}
Moreover, the regime $P\leq \min\{N_t,N_d-N_t^{\varepsilon_{\rm rel}}\}$ indicates the possible range $P\in [1,\Theta(N_d))$.
Combining this with Eq.~\eqref{eq:apdx_advlb_nontrivial} yields ${\rm Q}=\Omega(\varepsilon_{\rm rel}^{-1}/P)$.

\clearpage

\clearpage
\bibliography{main}

\clearpage
\onecolumngrid

\section*{Supplemental Material}
\setcounter{section}{0}
\renewcommand{\thesection}{S\arabic{section}}
\setcounter{subsection}{0}
\renewcommand{\thesubsection}{\Alph{subsection}}
\setcounter{equation}{0}
\renewcommand{\theequation}{S\arabic{equation}}
\setcounter{figure}{0}
\renewcommand{\thefigure}{S\arabic{figure}}
\setcounter{table}{0}
\renewcommand{\thetable}{S\arabic{table}}
\makeatletter
\renewcommand{\p@subsection}{\thesection\ }
\makeatother

\section{Comparison with other QAE}
\label{sec_comparison_QAE}
Table~\ref{table_QAE_comparison} summarizes the comparison of PAE with prior QAEs in terms of the number of qubits, maximum circuit depth, and query complexity.
Here, $n$ denotes the number of qubits on which $U_a$ acts.
$\varepsilon_{\rm add}$ represents the additive error, whereas $\varepsilon$ denotes the root mean squared error (RMSE).

\begin{table}[H]
    \label{table_QAE_comparison}
    \centering
    \small
    \begin{tabular}{|c|c|c|c|}
        \hline
        \parbox[c][0.6cm][c]{0cm}{} \textbf{Algorithm} & \textbf{\#qubits} & \textbf{Max~depth} & \textbf{\#Query}\\
        \hhline{|====|} 
        \parbox[c][1.2cm][c]{0cm}{} QAE \cite{brassard2002quantum} & $n + \mathcal{O}(\log(1/\varepsilon_{\rm add}))$     & $\mathcal{O}\left( \cfrac{1}{\varepsilon_{\rm add}}\right)$& $ \mathcal{O}\left(\cfrac{1}{\varepsilon_{\rm add}}\right)$ \\
        \hline
        \parbox[c][1.2cm][c]{0cm}{} MLAE \cite{suzuki2020amplitude}                  & $n$               & $\mathcal{O}\left(  \cfrac{1}{\varepsilon} \right)$             & $ \mathcal{O}\left(\cfrac{1}{\varepsilon}\right)$ \\
        \hline
        \parbox[c][1.2cm][c]{0cm}{} IQAE \cite{grinko2021iterative}                  & $n$               & $\mathcal{O}\left(  \cfrac{1}{\varepsilon_{\rm add}} \right)$  & $\mathcal{O}\left( \cfrac{1}{\varepsilon_{\rm add}}\right)$ \\
        \hline
        \parbox[c][1.2cm][c]{0cm}{} Power-law AE \cite{giurgica2022low}              & $n$               & $\mathcal{O}\left( \cfrac{1}{\varepsilon_{\rm add}^{{1-\kappa}}} \right)$  & $\tilde{\mathcal{O}}\left(\cfrac{1}{\varepsilon_{\rm add}^{1+\kappa}} \right)$ \\
        \hline
        \parbox[c][1.2cm][c]{0cm}{} PAE(general) [This work]     & $P(n+1)$     & $\mathcal{O}\left( \cfrac{1}{P\varepsilon} + \log P \right)$  & $\mathcal{O}\left(\cfrac{1}{\varepsilon} + P\log P \right)$ \\
        \hline
        \parbox[c][1.2cm][c]{0cm}{} PAE(fully parallel) [This work]     & $\mathcal{O}(n/\varepsilon)$     & $\mathcal{O}\left( \log\left(\cfrac{1}{\varepsilon}\right) \right)$  & $\mathcal{O}\left(\cfrac{\log(1 / \varepsilon)}{\varepsilon}\right)$ \\
        \hline
    \end{tabular}
    \caption{A comparison of QAE algorithms in terms of the number of qubits,   maximum circuit depth, and query complexity.
            Here, $n$ denotes the number of qubits acted on by $U_a$, $\varepsilon_{\rm add}$ represents the additive estimation error, and $\varepsilon$ denotes the RMSE.
            For the methods~\cite{brassard2002quantum,grinko2021iterative,giurgica2022low}, the complexity is evaluated such that the final estimate has an additive error $\varepsilon_{\rm add}$ in a high probability.
            The parameter $P$ is the degree of the parallelization in PAE, and $\kappa \in (0,1)$ controls the trade-off between circuit depth and query complexity in power-law AE.
            For simplicity, we ignore a log-log factor in IQAE.
            }
\end{table}

The ``fully parallel'' variant of PAE achieves $\mathcal{O}(\log(1/\varepsilon))$-depth at the cost of increased qubit resources.
To the best of our knowledge, PAE is the only method that achieves logarithmic depth scaling while maintaining query complexity at nearly Heisenberg limit (HL) scaling uniformly for all  $a \in [0, 1]$.

\section{Construction of engineered phase shifter with QSP}
\label{sec_operator_transformation_detail}

In this section, we explain how to construct the engineered phase shifter $V_{\varphi, T}$ using quantum signal processing (QSP).
First, we briefly review QSP, and then describe the procedure for constructing $V_{\varphi, T}$.

\subsection{Overview of QSP}
Given a unitary operator $W = \sum_\lambda e^{i\theta_\lambda}\ket{\lambda}\bra{\lambda}_s$ with $\ket{\lambda}$ the eigenstate of $W$ and $e^{i\theta_{\lambda}}$ the corresponding eigenvalue, QSP realizes a transformation of the eigenphases $\theta_{\lambda}$, by interleaving applications of the controlled operator $W_x$:
\begin{align}
\label{eq_def of Wx}
    W_x := \ket{+}\bra{+}_b \otimes \bm{1}_s + \ket{-}\bra{-}_b \otimes W
         = \sum_\lambda e^{i\theta_\lambda/2} R_x(\theta_\lambda)\otimes \ket{\lambda}\bra{\lambda}_s,
\end{align}
and $R_z(\xi) = e^{-i\xi Z_b / 2}$ \cite{low2017optimal, low2017quantum, low2019hamiltonian}.
This results in the operator $V_{x, \vec{\xi}}$:
\begin{align}
    \label{eq_QSP_V_Wx}
    V_{x, \vec{\xi}}   
    &= V_{x, \xi_{1}+\pi}^{\dagger} V_{x, \xi_{2}} V_{x, \xi_{3}+\pi}^{\dagger}
        \cdots V_{x, \xi_{L-1}+\pi}^{\dagger} V_{x, \xi_L} \nonumber \\
                    &= \sum_\lambda \left( \mathcal{A}(\theta_\lambda) \bm{1}_b
                                            +i\mathcal{B}(\theta_\lambda) Z_b 
                                            +i\mathcal{C}(\theta_\lambda) X_b 
                                            +i\mathcal{D}(\theta_\lambda) Y_b 
                                        \right) \otimes  \ket{\lambda}\bra{\lambda}_s \nonumber \\
                    &= \sum_\lambda \begin{pmatrix}
                                        \mathcal{A}(\theta_\lambda) + i\mathcal{B}(\theta_\lambda)  & i\mathcal{C}(\theta_\lambda) + \mathcal{D}(\theta_\lambda) \\
                                        i\mathcal{C}(\theta_\lambda) - \mathcal{D}(\theta_\lambda)  & \mathcal{A}(\theta_\lambda) - i\mathcal{B}(\theta_\lambda)
                                    \end{pmatrix}_b
                                    \otimes  \ket{\lambda}\bra{\lambda}_s, \\
    V_{x, \xi} & = \left( R_z(\xi) \otimes \bm{1}_s \right) W_x \left( R_z(-\xi) \otimes \bm{1}_s \right)
\end{align}
where $L$ is an even integer; this alternating product of $V_{x, \xi}$ and $V^\dagger_{x, \xi+\pi}$ uncomputes the unnecessary phase $e^{i\theta_\lambda/2}$ in 
Eq.~\eqref{eq_def of Wx}. 
Also, $\bm{1}_s$ and $\bm{1}_b$ denote the identity operators on the system and ancilla qubits, respectively.
$X_b,Y_b,Z_b$ are the Pauli operators acting on the ancilla qubit.
$\mathcal{A}, \mathcal{B}, \mathcal{C}$ and $\mathcal{D}$ are real-valued functions determined by the rotation angles $\vec{\xi}=(\xi_1,...,\xi_L)$.
The 2×2 matrix in the rightmost equation acts on the computational basis $\{\ket{0}_b, \ket{1}_b\}$ of the ancilla qubit.

The above construction of QSP using $W_x$ and $R_z(\xi)$ is referred to as the Wx-convention.
An alternative form, known as the Wz-convention~~\cite{chao2020finding, martyn2021grand}, uses the operator $W_z$:
\begin{align}
    \label{eq_Wz}
    W_z := \ket{0}\bra{0}_b \otimes \bm{1}_s + \ket{1}\bra{1}_b \otimes W
       = \sum_\lambda e^{i\theta_\lambda/2} R_z(\theta_\lambda)\otimes \ket{\lambda}\bra{\lambda}_s,
\end{align}
and $R_x(\xi)$, to construct 
\begin{equation}
    V_{z, \vec{\xi}} = V_{z, \xi_{1}+\pi}^{\dagger} V_{z, \xi_{2}} V_{z, \xi_{3}+\pi}^{\dagger}
        \cdots V_{z, \xi_{L-1}+\pi}^{\dagger} V_{z, \xi_L}, ~~~ 
    V_{z, \xi} = \left( R_x(\xi) \otimes \bm{1}_s \right) W_z \left( R_x(-\xi) \otimes \bm{1}_s \right). 
\end{equation}
Since the Wz-convention is well suited for constructing $V_{\varphi, T}$ that induces a relative phase $e^{iT\varphi}$ between $\ket{0}$ and $\ket{1}$, we employ this convention.
In the Wz-convention, the operator $V_{z, \vec{\xi}}$ is related to $V_{x, \vec{\xi}}$ in the following form~\cite{martyn2021grand}:
\begin{align}
    \label{eq_Vz}
    V_{z, \vec{\xi}}&= H_b V_{x,\vec{\xi}} H_b \nonumber \\
                    &= \sum_\lambda \begin{pmatrix}
                                        \mathcal{A}(\theta_\lambda) + i\mathcal{C}(\theta_\lambda)  & i\mathcal{B}(\theta_\lambda) - \mathcal{D}(\theta_\lambda) \\
                                        i\mathcal{B}(\theta_\lambda) + \mathcal{D}(\theta_\lambda)  & \mathcal{A}(\theta_\lambda) - i\mathcal{C}(\theta_\lambda)
                                    \end{pmatrix}_b
                                    \otimes  \ket{\lambda}\bra{\lambda}_s,
\end{align}
where $H_b$ is the Hadamard gate acting on the ancilla qubit.
As will be shown later,  to realize the required transformation, we only need to focus on $(\mathcal{A}, \mathcal{C})$.
The theorem below gives a complete characterization of $(\mathcal{A}, \mathcal{C})$.

\begin{theorem}[\bf Achievable $(\mathcal{A}, \mathcal{C})$ in QSP - Theorem 1 of \cite{low2017optimal}]
\label{thm_QSP}
For all even integers $L>0$,  a pair of real functions $\mathcal{A}, \mathcal{C}$ can be implemented by some angle sequence $\vec{\xi}\in \mathbb{R}^L$ if and only if the following conditions are satisfied:\\
(1) $\forall \theta \in \mathbb{R}, \mathcal{A}^2(\theta) + \mathcal{C}^2(\theta) \le 1$ \\
(2) $\mathcal{A}(0) = 1$ \\
(3) $\mathcal{A}(\theta) = \Sigma^{L/2}_{l=0} a_l\cos{(l\theta)} ,\ \{ a_l \} \in \mathbb{R}^{L/2 + 1}$ \\ 
(4) $\mathcal{C}(\theta) = \Sigma^{L/2}_{l=0} c_l\sin{(l\theta)} ,\ \{ c_l \} \in \mathbb{R}^{L/2}$ \\ 
Moreover, $\vec{\xi}$ can be computed from $\mathcal{A}(\theta)$ and $\mathcal{C}(\theta)$ in classical time $\mathcal{O}({\rm poly}(L))$.
\end{theorem}

\subsection{Detail of operator transformation with QSP}
We now detail the construction of the approximate phase shifter $V_{\varphi, T}$ using QSP.
As the operator $W_z$, we employ a slight modification of 
${\rm c}Q=\ket{0}\bra{0}\otimes \bm{1}_s + \ket{1}\bra{1}\otimes Q$ as follows:
\begin{align}
    \label{eq_def_Wq}
    W_Q &:= {\rm c}Q \times (R_z(\pi/2) \otimes \bm{1}_s) \nonumber \\
        &= e^{-i\pi/4}{e^{i(-2\theta + \pi/2)/2}}
        \begin{pmatrix}
            e^{-i(-2\theta + \pi/2)/2} & 0 \\
            0 & e^{i(-2\theta + \pi/2)/2} \\
        \end{pmatrix}_b
        \otimes \ket{Q_+}\bra{Q_+}_s \nonumber \\
        &\quad  + e^{-i\pi/4}{e^{i(2\theta + \pi/2)/2}}
        \begin{pmatrix}
            e^{-i(2\theta + \pi/2)/2} & 0 \\
            0 & e^{i(2\theta + \pi/2)/2} \\
        \end{pmatrix}_b
        \otimes \ket{Q_-}\bra{Q_-}_s,
\end{align}
where we used the expression \eqref{eq_cQ expression} and omit all terms that act outside of the Grover plane. 
Note that the factor $R_z(\pi/2)$ is multiplied to ${\rm c}Q$ so that the transformation functions satisfy the conditions in Theorem~\ref{thm_QSP}. 
To approximate $\widetilde{V}_{\varphi, T}$ defined in Eq.~\eqref{eq_def_Vphi}, i.e., 
\begin{equation}
    \label{eq_def_Vphi suppl}
    \widetilde{V}_{\varphi, T}=\begin{pmatrix}
        e^{-iT\varphi/2}&0\\
        0&e^{iT\varphi/2}
    \end{pmatrix}_b \otimes \overline{\bm{1}}_s,
\end{equation}
we use QSP to construct $V_{\varphi, T} = V_{z, \vec{\xi}}$ of the form:
\begin{align}
\label{eq_def_Vph_detail 2}
    V_{\varphi, T} &= \prod^{L/2}_{l = 1} \left( R_x(\xi_{2l-1}') \otimes \bm{1}_s \right) W_Q^{\dagger} \left( R_x(-\xi_{2l-1}') \otimes \bm{1}_s \right)  \left( R_x(\xi_{2l}) \otimes \bm{1}_s \right) W_Q \left( R_x(-\xi_{2l}) \otimes \bm{1}_s \right) \nonumber \\
      &= \sum_{\sigma \in \{ +, - \}}
          \begin{pmatrix}
                     \mathcal{A}_T(\theta_{Q_\sigma}) + i\mathcal{C}_T(\theta_{Q_\sigma})  & i\mathcal{B}_T(\theta_{Q_\sigma}) - \mathcal{D}_T(\theta_{Q_\sigma}) \\
                     i\mathcal{B}_T(\theta_{Q_\sigma}) + \mathcal{D}_T(\theta_{Q_\sigma})  & \mathcal{A}_T(\theta_{Q_\sigma}) - i\mathcal{C}_T(\theta_{Q_\sigma})
          \end{pmatrix}_b
         \otimes  \ket{Q_\sigma}\bra{Q_\sigma}_s,
\end{align}
where $\sigma \in \{ +, - \}$ and $\theta_{Q_\pm} := \mp 2\theta + \pi/2$.
Also, $\vec{\xi}$ is a QSP hyperparameter and $\xi' = \xi + \pi$. 
The circuit structure of $V_{\varphi, T}$ is shown in Fig.~\ref{fig_circuit_Vph}.
Hence, to approximate $\widetilde{V}_{\varphi, T}$, it suffices to construct $V_{\varphi, T}$ such that
\begin{align}
    \mathcal{A}_T(\theta_{Q_\sigma}) \pm i\mathcal{C}_T(\theta_{Q_\sigma}) 
    = e^{\mp iT\sin{\theta_{Q_\sigma}}} \nonumber = e^{\mp iT\cos{(2 \theta})}.
\end{align}
As shown in Ref.~\cite{low2017optimal}, $e^{\mp iT\sin{\theta_{Q_\sigma}}}$ can be expressed via the Jacobi--Anger expansion:
\begin{align}
    e^{\mp iT\sin{\theta_{Q_\sigma}}} &= J_0(T) + 2\sum^{\infty}_{l \ {\rm even} > 0} J_l(T) \cos{(l\theta_{Q_\sigma})} \mp 2i \sum^{\infty}_{l \ {\rm odd} > 0} J_l(T) \sin{(l\theta_{Q_\sigma})},
\end{align}
where $J_l(T)$ denotes the Bessel function of the first kind of order $l$.
We define $\widetilde{A}_T(\theta_{Q_\sigma})$ and $\widetilde{C}_T(\theta_{Q_\sigma})$ as follows:
\begin{align}
    \label{eq_QSP_ACwidetilde}
    \widetilde{A}_T(\theta_{Q_\sigma}) = J_0(T) + 2\sum^{L/2}_{l \ {\rm even} > 0} J_l(T) \cos{(l\theta_{Q_\sigma})}, \quad
    i\widetilde{C}_T(\theta_{Q_\sigma}) = 2i \sum^{L/2}_{l \ {\rm odd} > 0} J_l(T) \sin{(l\theta_{Q_\sigma})}. 
\end{align}
Although $\widetilde{A}_T$ and $\widetilde{C}_T$ may not satisfy the conditions (1) and/or (2) in Theorem~\ref{thm_QSP}, they can be approximated by some functions $A_T$ and $C_T$ which satisfy these conditions~\cite{low2017optimal, low2019hamiltonian}.
Therefore, we can construct $V_{\varphi, T}$ in Eq.~\eqref{eq_def_Vph_detail 2} such that $\mathcal{A}_T(\theta_{Q_\sigma}) \pm i\mathcal{C}_T(\theta_{Q_\sigma}) ={A}_T(\theta_{Q_\sigma}) \pm i{C}_T(\theta_{Q_\sigma}) \approx e^{\mp iT\sin{\theta_{Q_\sigma}}}= e^{\mp iT\cos{(2 \theta})}$. 
We can control this approximation error by adjusting the truncation number $L$.

Based on the above discussion, we can take a pair of real functions $\mathcal{A}=A_T$, $\mathcal{C}=C_T$ satisfying all the conditions in Theorem~\ref{thm_QSP} and the following approximation
\begin{equation}
    \label{eq_AC-error}
    \left|A_T(\theta)\pm i C_T(\theta)-e^{\mp iT \sin(\theta)}\right|\leq \mathcal{O}(\delta)~~~\forall \theta
\end{equation}
for an error parameter $\delta$, which is explicitly specified later.
Thus, there exists a phase sequence $\vec{\xi}$ for the function pair $(A_T,C_T)$.
Under this choice, the corresponding functions $\mathcal{B}=B_T$ and $\mathcal{D}=D_T$ satisfy
\begin{equation}
    \label{eq_BD-error}
    \left|i{B}_T(\theta)\pm{D}_T(\theta)\right|^2= 1-\left|A_T(\theta)\pm i C_T(\theta)\right|^2\leq \mathcal{O}(\delta),
\end{equation}
where the first equality comes from the unitarity of any QSP circuit~\cite{low2017quantum}.
Therefore, our QSP circuit $V_{\varphi, T}$ with interleaving applications of controlled Grover operator ${\rm c}Q$ (more precisely, $W_Q$ in Eq.~\eqref{eq_def_Wq}) has the following action in the Grover plane:
\begin{align}
    \label{eq_def_Vph_detail_supp}
    V_{\varphi, T} &={\prod^{L/2}_{l = 1} \left( R_x(\xi_{2l-1} + \pi) \otimes \bm{1}_s \right) W_Q^{\dagger} \left( R_x(-\xi_{2l-1} - \pi) \otimes \bm{1}_s \right) \left( R_x(\xi_{2l}) \otimes \bm{1}_s \right) W_Q \left( R_x(-\xi_{2l}) \otimes \bm{1}_s \right)} \nonumber \\
    &= \sum_{\sigma \in \{ +, - \}} \begin{pmatrix}
                                            A_T(\theta_{Q_\sigma}) + iC_T(\theta_{Q_\sigma})  & iB_T(\theta_{Q_\sigma}) - D_T(\theta_{Q_\sigma}) \\
                                            iB_T(\theta_{Q_\sigma}) + D_T(\theta_{Q_\sigma})  & A_T(\theta_{Q_\sigma}) - iC_T(\theta_{Q_\sigma})
                                        \end{pmatrix}_b
                                        \otimes  \ket{Q_\sigma}\bra{Q_\sigma}_s \nonumber \\ 
    &\approx \sum_{\sigma \in \{ +, - \}}
        \begin{pmatrix}
            e^{-iT\sin (\theta_{Q_\sigma})} & 0 \\
            0 & e^{iT\sin (\theta_{Q_\sigma})} \\
        \end{pmatrix}_b
        \otimes \ket{Q_\sigma}\bra{Q_\sigma}_s \nonumber \\
    &= 
        \begin{pmatrix}
            e^{-iT\cos (2\theta)} & 0 \\
            0 & e^{iT\cos (2\theta)} \\
        \end{pmatrix}_b
        \otimes \sum_{\sigma \in \{ +, - \}} \ket{Q_\sigma}\bra{Q_\sigma}_s \nonumber \\
    &= 
        \begin{pmatrix}
            e^{-iT\varphi/2} & 0 \\
            0 & e^{iT\varphi/2} \\
        \end{pmatrix}_b
        \otimes \overline{\bm{1}}_s ~ =\widetilde{V}_{\varphi, T},
\end{align}
where $\varphi = 2\cos{(2\theta)} = 2(1-2a)$, and terms acting outside the Grover plane are again omitted.
The approximation in the third line comes from Eqs.~\eqref{eq_AC-error} and~\eqref{eq_BD-error}.
Here, $\overline{\bm{1}}_s$ is the identity operator on the Grover plane and has the spectral decomposition $\overline{\bm{1}}_s = \sum_{\sigma \in \{ +, - \}} \ket{Q_\sigma}\bra{Q_\sigma}_s$ for the orthogonal basis set $\{\ket{Q_+},\ket{Q_-}\}$ in the Grover plane.
Eq.~\eqref{eq_def_Vph_detail_supp} shows that $\widetilde{V}_{\varphi, T}$ can be implemented only approximately with a controllable accuracy $\delta$; we will derive how the number $L$ scales in the approximation error $\delta$ in Sec.~\ref{sec_proof_oc}.

\section{Proof of query complexity for constructing $V_{\varphi, T}$}
\label{sec_proof_oc}

Here, we provide the detailed proof of Lemma~\ref{lm_operator_transformation}, showing the error between $V_{\varphi, T}$ and $\widetilde{V}_{\varphi, T}$ on specific vectors $\ket{0}_b \ket{0}^{\otimes n}_s$ and $\ket{1}_b \ket{0}^{\otimes n}_s$.
We also discuss the effect of the approximation error in $V_{\varphi, T}$ when it is applied $S$ times sequentially instead of increasing $T$. 

\setcounter{lemma}{0}
\begin{lemma}
[Query complexity for constructing $V_{\varphi, T}$]
For any oracle conversion error $\varepsilon_{\rm oc}\in (0, 1)$ and 
any $j \in \{0, 1\}$,
there exists a quantum algorithm that constructs an operator $V_{\varphi, T}$ such that 
\begin{equation}
\label{lemma1 ineq}
   \Big\lVert \left( V_{\varphi, T} - \widetilde{V}_{\varphi, T} \right) \ket{j}_b \ket{0}^{\otimes n}_s \Big\rVert < \varepsilon_{\rm oc},
\end{equation}
using ${\rm c}Q$ and ${\rm c}Q^{\dagger}$ a total of $L = \mathcal{O}(T + \log(1/\varepsilon_{\rm oc}))$ times.
\end{lemma}

\mbox{}

\begin{proof}[Proof of Lemma~\ref{lm_operator_transformation}] 
As shown in Sec.~\ref{sec_operator_transformation_detail}, using QSP, we can construct $V_{\varphi, T}$ that approximates $\widetilde{V}_{\varphi, T}$.
The construction of $V_{\varphi, T}$ involves two approximations. 
The first is the approximation of $e^{\mp iT\sin{\theta_{Q_\sigma}}}$ by the $L/2$-order Fourier series $\widetilde{A} \pm i\widetilde{C}$ defined in Eq.~\eqref{eq_QSP_ACwidetilde}.
The error caused by this approximation is upper-bounded as follows for any $\theta_{Q_\sigma}$~\cite{low2017optimal}:
\begin{align}
    \left| \widetilde{A}_T(\theta_{Q_\sigma}) \pm i\widetilde{C}_T(\theta_{Q_\sigma}) - e^{\mp iT\sin{\theta_{Q_\sigma}}} \right| =: \delta  &\le \cfrac{4T^{L/2 + 1}}{2^{L/2 + 1}(L/2 + 1)!} \label{eq_delta_1} \\
    &< \cfrac{4}{e^{1/(6L+13)}\sqrt{2\pi(L/2 + 1)}}  \left( \cfrac{eT}{L + 2} \right)^{L/2+1} \nonumber \\
    &< 1.1 \left( \cfrac{eT}{L + 2} \right)^{L/2+1}, \label{eq_delta_2} 
\end{align}
where $\sigma \in \{+, -\}$, and we used Stirling's approximation $(L/2 + 1)! > e^{1/(6L+13)}\sqrt{2\pi(L/2 + 1)}\left( \frac{L/2+1}{e} \right)^{L/2+1}$.
In the following discussion, we assume $\delta \in (0, 1)$ and $\left( eT/(L+2) \right)^{L/2+1}\in (0, 1]$.
The second approximation is the replacement of $\widetilde{A}$ and $\widetilde{C}$ with the achievable functions $A$ and $C$ that satisfy all the conditions in Theorem~\ref{thm_QSP}. 
Using the technique shown in Ref.~\cite{low2017optimal}, we can construct such $A$ and $C$ satisfying the following inequality for any $\theta_{Q_\sigma}$~\cite{low2017optimal}, in terms of the definition of $\delta$ given in Eq.~\eqref{eq_delta_1}:
\begin{align}
\label{eq_8 delta}
    \left| A_T(\theta_{Q_\sigma}) \pm iC_T(\theta_{Q_\sigma}) - e^{\mp iT\sin{\theta_{Q_\sigma}}} \right| &\le 8\delta.
\end{align}

Here, we express $V_{\varphi, T}$ as follows:
\begin{align}
\label{eq_V FG expression}
    V_{\varphi, T}  &= \sum_{\sigma \in \{ +, - \}}   \begin{pmatrix}
                                            A_T(\theta_{Q_\sigma}) + iC_T(\theta_{Q_\sigma})  & iB_T(\theta_{Q_\sigma}) - D_T(\theta_{Q_\sigma}) \\
                                            iB_T(\theta_{Q_\sigma}) + D_T(\theta_{Q_\sigma})  & A_T(\theta_{Q_\sigma}) - iC_T(\theta_{Q_\sigma})
                                        \end{pmatrix}_b
                                        \otimes  \ket{Q_\sigma}\bra{Q_\sigma}_s \nonumber \\
          &:= \sum_{\sigma \in \{ +, - \}}   \begin{pmatrix}
                                            \mathcal{F}_{0, \sigma, T}  & i\mathcal{G}_{0, \sigma, T} \\
                                            i\mathcal{G}_{1, \sigma, T} & \mathcal{F}_{1, \sigma, T} \\
                                        \end{pmatrix}_b 
                                        \otimes \ket{Q_\sigma} \bra{Q_\sigma}_s.
\end{align}
Then, from Eqs. \eqref{eq_def_Vphi suppl} and \eqref{eq_V FG expression}, 
the following inequality holds:
\begin{align}
    \label{eq_Vph_norm-error_Qsigma}
    \left\lVert(V_{\varphi, T} - \widetilde{V}_{\varphi, T})\ket{j}_b \ket{Q_\sigma}_s \right\rVert 
    &= \left\lVert (\mathcal{F}_{j, \sigma, T} - e^{(-1)^{j+1}iT\sin{\theta_{Q_\sigma}}}) \ket{j}_b\ket{Q_\sigma}_s + i\mathcal{G}_{j', \sigma, T} \ket{j'}_b\ket{Q_\sigma}_s \right\rVert \nonumber \\
    &= \left\lVert (\mathcal{F}_{j, \sigma, T} - e^{(-1)^{j+1}iT\sin{\theta_{Q_\sigma}}}) \ket{j}_b + i\mathcal{G}_{j', \sigma, T} \ket{j'}_b \right\rVert \nonumber \\
    &= \sqrt{\left| \mathcal{F}_{j, \sigma, T} - e^{i T\phi_j} \right|^2 + \left| \mathcal{G}_{j', \sigma, T} \right|^2} \nonumber \\
    &\le \left| \mathcal{F}_{j, \sigma, T} - e^{i T\phi_j} \right| + \left| \mathcal{G}_{j', \sigma, T} \right| \nonumber \\
    &< 8\delta + \sqrt{16\delta - 64 \delta^2},
\end{align}
where $j' \in \{0, 1\}, j' \ne j$, and $(-1)^{j+1}\sin{\theta_{Q_\sigma}} = (-1)^{j+1}\cos{2\theta} := \phi_j$.
To derive the rightmost inequality, we used Eq.~\eqref{eq_8 delta}, i.e., 
$\left| \mathcal{F}_{j, \sigma, T} - e^{i T\phi_j} \right| \le 8\delta$, and 
$\left| \mathcal{F}_{j, \sigma, T} \right|^2 + \left| \mathcal{G}_{j', \sigma, T} \right|^2 = 1$, which further lead to
\begin{align}
    &1 - \left| \mathcal{F}_{j, \sigma, T} \right| \le \left| \mathcal{F}_{j, \sigma, T} - e^{i T\phi_j} \right| \le 8\delta ~~
    \Longrightarrow ~~ 1 - 8\delta \le \left| \mathcal{F}_{j, \sigma, T} \right| ~~ \Longrightarrow ~~ \left| \mathcal{G}_{j', \sigma, T} \right| \le \sqrt{16\delta - 64 \delta^2}.
\end{align}
From Eq.~\eqref{eq_Vph_norm-error_Qsigma}, the following inequality holds:
\begin{align}
    \left\lVert (V_{\varphi, T} - \widetilde{V}_{\varphi, T})\ket{j}_b \ket{0}^{\otimes n}_s  \right\rVert 
    &\le \cfrac{1}{\sqrt{2}} \left( \left\lVert(V_{\varphi, T} - \widetilde{V}_{\varphi, T}) \ket{j}_b \ket{Q_+}_s  \right\rVert + \left\lVert(V_{\varphi, T} - \widetilde{V}_{\varphi, T}) \ket{j}_b \ket{Q_{-}}_s\right\rVert \right) \nonumber \\
    &< \cfrac{2}{\sqrt{2}} \left( 8\delta + \sqrt{16\delta - 64 \delta^2} \right) \label{eq_Vph_norm-error_1} \\ 
    &< 17 \sqrt{\delta}. \label{eq_Vph_norm-error_2}
\end{align}
Based on Eqs.~\eqref{eq_delta_2}, \eqref{eq_Vph_norm-error_2}, we then have 
\begin{equation*}
     17\sqrt{\delta} < 17\sqrt{1.1}\left( \cfrac{eT}{L + 2} \right)^{L/4+1/2} 
       < 18\left( \cfrac{eT}{L + 2} \right)^{L/4+1/2}.
\end{equation*}
Hence, to ensure that $\big\lVert(V_{\varphi, T} - \widetilde{V}_{\varphi, T})\ket{j}_b \ket{0}^{\otimes n}_s  \big\rVert \le \varepsilon_{\rm oc}$, it suffices that $18\left( eT/(L+2) \right)^{L/4+1/2} \le \varepsilon_{\rm oc}$. 
According to Ref.~\cite{gilyen2019quantum}, this inequality is satisfied when 
\begin{align}
    L = 2 \left\lceil \cfrac{e^2 T + 4 \log{(18/\varepsilon_{\rm oc})}-2}{2} \right\rceil. 
\end{align}
Therefore, by setting
\begin{align}
\label{eq_L}
    L = 2 \left\lceil \cfrac{e^2 T + 4 \log(1/\varepsilon_{\rm oc}) + 10}{2} \right\rceil = \mathcal{O}(T + \log(1/\varepsilon_{\rm oc})),
\end{align}
the inequality $\big\lVert ( V_{\varphi, T} - \widetilde{V}_{\varphi, T} ) \ket{j}_b \ket{0}^{\otimes n}_s \big\rVert < \varepsilon_{\rm oc}$ holds.

\end{proof}

We now provide a proof of the following inequality presented in Appendix~\ref{sec_pseudocode}: if Eq.~\eqref{lemma1 ineq} holds, then for any positive integer $S$ we have
\begin{align}
    \Big\lVert \left( V_{\varphi, T}^S - \widetilde{V}_{\varphi, T}^S \right) \ket{j}_b \ket{0}^{\otimes n}_s  \Big\rVert < S\varepsilon_{\rm oc}.
\end{align}

\begin{proof}
\begin{align}
    \left\lVert (V_{\varphi, T}^S - \widetilde{V}_{\varphi, T}^S) \ket{j}_b \ket{0}^{\otimes n}_s  \right\rVert &= \left\lVert \sum_{k=0}^{S-1} V_{\varphi, T}^{S-k-1} \left( V_{\varphi, T} - \widetilde{V}_{\varphi, T} \right) \widetilde{V}_{\varphi, T}^{k} \ket{j}_b \ket{0}^{\otimes n}_s  \right\rVert \label{eq_V^S_error} \\
    &\le \sum_{k=0}^{S-1} \left\lVert V_{\varphi, T}^{S-k-1} \left( V_{\varphi, T} - \widetilde{V}_{\varphi, T} \right) \widetilde{V}_{\varphi, T}^{k} \ket{j}_b \ket{0}^{\otimes n}_s  \right\rVert \nonumber \\
    &= \sum_{k=0}^{S-1} \left\lVert \left( V_{\varphi, T} - \widetilde{V}_{\varphi, T} \right) e^{ik T\phi_j} \ket{j}_b \ket{0}^{\otimes n}_s  \right\rVert \nonumber \\
    &= S \times \left\lVert (V_{\varphi, T} - \widetilde{V}_{\varphi, T})\ket{j}_b \ket{0}^{\otimes n}_s  \right\rVert \nonumber \\
    &< S\varepsilon_{\rm oc}.
\end{align}
\end{proof}

\section{Classical post-processing in robust phase estimation}
\label{sec_RPE}
We describe the classical post-processing procedure of robust phase estimation (RPE)~\cite{higgins2009demonstrating, kimmel2015robust, belliardo2020achieving}.
Throughout this section, we assume $\varphi \in [-\pi, \pi)$ to describe the general RPE procedure, whereas PAE assumes $\varphi \in [-2, 2]$.
Given the quantum circuit measurement outcomes $\{f_{+,k}\}_{k=1}^K$ and $\{f_{i,k}\}_{k=1}^K$, the following procedure is executed for $k = 1, 2, ..., K$ to estimate $\varphi$.
Hereafter, $\varphi'$ and $\widehat{\varphi}'$ denote the values of $\varphi$ and $\widehat{\varphi}$ mapped from $[-\pi, \pi)$ to $[0, 2\pi)$.

\begin{enumerate}
    \item Derive estimate $\widehat{M_k\varphi'_{k}} := {\rm atan2}(2f_{i, k} - 1, 2f_{+, k} - 1) \in [0, 2\pi)$, where $M_k = 2^{k-1}$.
    \item Calculate $\widehat{\varphi}'_{k, 0} := \widehat{M_k\varphi'_{k}} / M_k  \in [0, 2\pi/M_k)$.
          As shown in Fig.~\ref{fig_RPE_phi_est}, $\widehat{\varphi}'_{k, 0}$ represents the smallest candidate for $\widehat{\varphi}'_k$ in the range $[0, 2\pi)$.
    \item If $k = 1$, adopt $\widehat{\varphi}'_{0, 1}$ as $\widehat{\varphi}'_{1}$. \par
          If $k > 1$, select $\widehat{\varphi}'_k$ from the candidate estimates $\{ \widehat{\varphi}'_{k, m} = \widehat{\varphi}'_{k, 0} + m\pi / 2^{k-2} \}_{m=-1}^{2^{k-1}}$  based on the previous estimate $\widehat{\varphi}'_{k-1}$.
          First, compute the partition index $\eta := \left\lfloor \widehat{\varphi}'_{k-1} / 2^{-k+2}\pi \right\rfloor \in \{0, 1,..., 2^{k-1}-1 \}$, which identifies the partition in which $\widehat{\varphi}'_{k-1}$ lies (see Fig.~\ref{fig_RPE_phi_est}).
          Then, select $\widehat{\varphi}'_{k}$ from the candidate estimates corresponding to the partition indices $\eta - 1$, $\eta$, and $\eta + 1$ whose confidence intervals, defined as the estimate $\pm \pi / 3 \times 2^{k-1}$, overlap with that of $\widehat{\varphi}'_{k-1}$.
    \item Map $\widehat{\varphi}'_{k}$ onto $[-\pi, \pi)$ to obtain $\widehat{\varphi}_{k}$.        
\end{enumerate}
The final estimate $\widehat{\varphi}$ is given by $\widehat{\varphi}_K$.
Figure~\ref{fig_RPE_phi_est} schematically illustrates the above estimation procedure.

\begin{figure}[H]
    \centering
    \includegraphics[width=\textwidth]{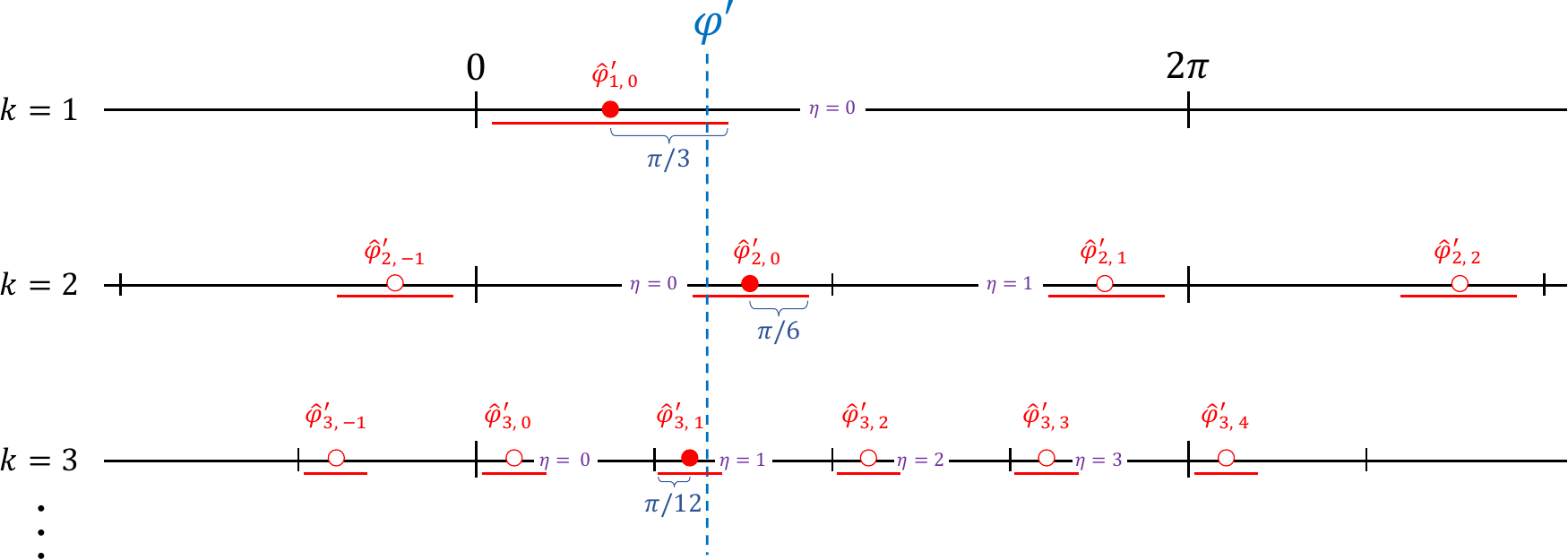}
    \caption{
        Schematic diagram of the step-by-step estimation of $\varphi$ in RPE.
        The filled red circles indicate the estimates adopted as $\widehat{\varphi}'_{k}$, that is, in this example, we have $\widehat{\varphi}'_{1} = \widehat{\varphi}'_{1, 0}$, $\widehat{\varphi}'_{2} = \widehat{\varphi}'_{2, 0}$ and $\widehat{\varphi}'_{3} = \widehat{\varphi}'_{3, 1}$.
    }
    \label{fig_RPE_phi_est}
\end{figure}

Below, we restate the pseudocode for the classical post-processing presented in Appendix~\ref{sec_pseudocode}.

\setcounter{algorithm}{1}
\begin{algorithm}[H]
    \caption{Robust phase estimation (classical post-processing part)}
    \label{alg_RPE}
    \begin{algorithmic}[1]
    \Require Max. number of steps $K$, Observed probabilities $\{f_{+,k}\}_{k=1}^K$, $\{f_{i,k}\}_{k=1}^K$
    \Ensure Estimate $\widehat{\varphi} \in [-\pi, \pi) $
    \For{$k = 1, 2,..., K$}
        \State $M_k=2^{k-1}$
        \State $\widehat{M_k\varphi_{k}'} \gets {\rm atan2}(2f_{i, k} - 1, 2f_{+, k} - 1) \in [0, 2\pi)$
        \State $\widehat{\varphi}'_{k, 0} = \widehat{M_k\varphi_{k}'} / M_k  \in [0, 2\pi/M_k)$.
        \If {$k = 1$}
            \State $\widehat{\varphi}'_1 \gets \widehat{\varphi}'_{1, 0}$
        \Else
            \State $\eta \gets \left\lfloor \cfrac{\widehat{\varphi}'_{k-1}}{\pi/2^{k-2}} \right\rfloor$
            \If{$\widehat{\varphi}'_{k-1} - \left( \widehat{\varphi}'_{k, 0} + (\eta - 1)\pi / 2^{k-2} \right) \le \pi/2^{k-1}$}
                \State $\widehat{\varphi}'_k \gets \widehat{\varphi}'_{k, 0} + (\eta - 1)\pi / 2^{k-2}$
            \ElsIf{$\left( \widehat{\varphi}'_{k, 0} + (\eta + 1)\pi / 2^{k-2} \right) - \widehat{\varphi}'_{k-1} < \pi/2^{k-1}$}
                \State $\widehat{\varphi}'_k \gets \widehat{\varphi}'_{k, 0} + (\eta + 1)\pi / 2^{k-2}$
            \Else
                \State $\widehat{\varphi}'_k \gets \widehat{\varphi}'_{k, 0} + \eta\pi / 2^{k-2}$
            \EndIf
        \EndIf
    \State $\widehat{\varphi}_k \gets \widehat{\varphi}'_k - 2\pi \left\lfloor \cfrac{\widehat{\varphi}'_k + \pi}{2\pi} \right\rfloor$
    \EndFor
    \State $\widehat{\varphi} \gets \widehat{\varphi}_{K}$    
    \end{algorithmic}
\end{algorithm}

\section{Full proof of theorem 1}
\label{sec_proof_theorem1}
Here, we show the full proof of Theorem~\ref{thm_PAE}.
For completeness, we first recall Lemma~\ref{lm_RPE_RMSE} from the main text, which is used in the proof, and then present the proof.

\begin{lemma}
[MSE upper bound of RPE~\cite{belliardo2020achieving}]
\label{lm_RPE_RMSE}
Suppose the measurement bias parameters $\{\beta_{r, k}\}$ satisfy $\sup_{r, k} \{|\beta_{r, k}|\} := \beta < \sqrt{6}/8$.
Then, the RPE procedure (i)--(ii) returns the phase estimate $\hat{\varphi} \in (-\pi, \pi]$ such that its mean squared error (MSE) satisfies
\begin{align}
    \label{eq_eps_upper-bound_supp}
    \mathbb{E}\left[ (\hat{\varphi} - \varphi)^2 \right] \le \left( \cfrac{2\pi}{3} \right)^2 \left( \cfrac{1}{4^K} + \sum_{k=1}^K \cfrac{e^{-2 \nu_k (\sqrt{6}/8 - \beta)^2}}{4^{k-4}} \right).
\end{align}
\end{lemma}

\setcounter{theorem}{0}
\begin{theorem}[Parallel amplitude estimation; general case]
    \label{thm_PAE}
    Let $\varepsilon\in (0,1)$, and let $P$ be any positive integer.
    There exists a quantum algorithm that estimates $a \in [0, 1]$ encoded in $U_a$ (Eq.~\eqref{eq:quantum_oracle_qae}) within the RMSE $\varepsilon$, using $N=\mathcal{O}\left(1/\varepsilon + P \log{P}\right)$ queries to $U_a$ and $U_a^\dagger$ in total.
    This quantum algorithm uses $P(n+1)$-qubit quantum circuits with the structure depicted in Fig.~\ref{fig_circuit_overview} and the maximum circuit depth of $\mathcal{O}(1/(\varepsilon P) + \log P)$.
\end{theorem}

\textit{Proof of Theorem~\ref{thm_PAE}.} 
The goal is, in the framework of RPE, to compute the necessary resources (circuit depth and width) such that the right hand side of Eq.~\eqref{eq_eps_upper-bound_supp} in Lemma~\ref{lm_RPE_RMSE} is at most $\varepsilon^2$. 
The condition in Lemma~\ref{lm_RPE_RMSE} on the measurement bias, $|\beta_{r, k}| < \beta$, is related to the approximation error of $V_{\varphi, T}$, which allows us to identify the necessary circuit depth from Lemma~\ref{lm_operator_transformation}. 
Hence, let us begin by evaluating the state error.

When $V_{\varphi, T_k}$ is applied in parallel to $P_k$ systems in the same manner as Fig.~\ref{fig_circuit_overview}, the following inequality holds by a telescoping sum:
\begin{align}
    \Big\lVert \left( V_{\varphi, T_k}^{\otimes P_k} - \widetilde{V}_{\varphi, T_k}^{\otimes P_k} \right) \ket{{\rm GHZ}_{P_k}}_b \ket{0}^{\otimes n P_k}_s \Big\rVert &\le \sqrt{2} \max_{j=0,1}\Big\lVert \left( V_{\varphi, T_k}^{\otimes P_k} - \widetilde{V}_{\varphi, T_k}^{\otimes P_k} \right) \ket{j}_b^{\otimes P_k} \ket{0}^{\otimes n P_k}_s \Big\rVert \nonumber \\
    &\le \sqrt{2} P_k \max_{j=0,1} \Big\lVert (V_{\varphi, T_k} - \widetilde{V}_{\varphi, T_k}) \ket{j}_b \ket{0}^{\otimes n}_s \Big\rVert.
\end{align}
In addition, for $\ket{\widetilde{\Psi}(M_k)} := \widetilde{V}_{\varphi, T_k}^{\otimes P_k} \ket{{\rm GHZ}_{P_k}}_b \ket{0}^{\otimes n P_k}_s$ and its approximation $\ket{\Psi(M_k)}$, we have 
\[
    \mathfrak{D}(\ket{\Psi(M_k)}, \ket{\widetilde{\Psi}(M_k)}) \le \lVert \ket{\Psi(M_k)} - \ket{\widetilde{\Psi}(M_k)} \rVert, 
\]
where $\mathfrak{D}(\ket{\Psi(M_k)}, \ket{\widetilde{\Psi}(M_k)})$ is the trace distance between these states.
Now, we connect this state error to the bias error $\beta_{r, k}$ in measuring the states $\ket{\Psi(M_k)}$ and $\ket{\widetilde{\Psi}(M_k)}$. 
Specifically, from the result that the trace distance between two quantum states upper bounds the total variation distance for any POVM \cite{nielsen2010quantum}, we have $|\beta_{r, k}| \le \mathfrak{D}(\ket{\Psi(M_k)}, \ket{\widetilde{\Psi}(M_k)})$.
Therefore, 
\begin{align}
    |\beta_{r, k}| &\le \sqrt{2} P_k \max_j \big\lVert (V_{\varphi, T_k} - \widetilde{V}_{\varphi, T_k}) \ket{j}_b \ket{0}^{\otimes n}_s \big\rVert \nonumber \\
    &< \sqrt{2} P_k \varepsilon_{\rm oc}.
\end{align}
According to Eq.~\eqref{eq_L}, by setting $\varepsilon_{\rm oc} = \beta / (\sqrt{2}P_k)$ and
\begin{align}
    \label{eq_Lk}
    L_k = 2 \left\lceil \cfrac{e^2 T_k + 4 \log(\sqrt{2}P_k/\beta) + 10}{2} \right\rceil, 
\end{align}
with $\beta \in (0, \sqrt{6}/8)$ in Lemma~\ref{lm_operator_transformation}, we have $|\beta_{r, k}| < \beta$. 
Hence, the MSE of $\varphi$ is upper bounded by Eq.~\eqref{eq_eps_upper-bound_supp} in Lemma~\ref{lm_RPE_RMSE}.

Next, we upper bound the MSE of $\varphi$.
Substituting 
\begin{align}
    K &= \lceil \log_2(1/\varepsilon) \rceil + 6, \\
    \label{eq_nuk}
    \nu_k &= 1 + \left\lceil \cfrac{\log{6}}{2(\sqrt{6}/8-\beta)^2} (K - k) \right\rceil, 
\end{align}
into the inequality in Eq.~\eqref{eq_eps_upper-bound_supp} yields the following inequality:
\begin{align}
    \mathbb{E}\left[ (\hat{\varphi} - \varphi)^2 \right] &\le \left( \cfrac{2\pi}{3} \right)^2 \left( \cfrac{1}{4^K} + \sum_{k=1}^K \cfrac{e^{-2 \nu_k (\sqrt{6}/8 - \beta)^2}}{4^{k-4}} \right) \nonumber \\
    &\le \left( \cfrac{2\pi}{3} \right)^2 \left( \cfrac{1}{4^K} + \sum_{k=1}^K \cfrac{e^{ (k - K)\log 6 - 2(\sqrt{6}/8 - \beta)^2 }}{4^{k-4}} \right) \nonumber \\
    &= \left( \cfrac{2\pi}{3} \right)^2 \left( \cfrac{1}{4^K} + \cfrac{768}{4^K}\left( 1 - \left(2/3\right)^{K} \right) e^{-2(\sqrt{6}/8 - \beta)^2} \right) \nonumber \\
    &< \cfrac{769}{4^K} \left( \cfrac{2\pi}{3} \right)^2 \nonumber \\ 
    &< \varepsilon^2.
\end{align}
Since $\varphi = 2(1-2a)$, we obtain $\mathbb{E}\left[ (\hat{a} - a)^2 \right] \le \mathbb{E}\left[ (\hat{\varphi} - \varphi)^2 \right]$ and thus $\sqrt{ \mathbb{E} \left[ (\hat{a} - a)^2 \right]} < \varepsilon$.

Then, we upper bound the total number $N$ of queries to the operator $U_a$ and $U_a^\dagger$.
Below, we consider $P \in \mathbb{Z}_{\ge 1}$ as an upper bound on the degree of the parallelism, and set $P_K = \min\{2^{K-1}, 2^{\lfloor \log_2 P \rfloor}\}$ (i.e., increase the degree of parallelism as much as possible).
Since the operator $V_{\varphi, T_k}$ contains a total of $L_k+2$ queries to $U_a$ and $U_a^\dagger$ as presented in Appendix~\ref{sec_operator_transformation}, total query $N$ for PAE is $N = 2 \sum_{k=1}^{K} \nu_k (L_k + 2) P_k$.
Here, the prefactor $2$ accounts for the two measurement settings $r\in\{+,\;i\}$ used in RPE.
By Eq.~\eqref{eq_Lk} and the RPE constraint $M_k = P_kT_k = 2^{k-1}$, we can upper bound $N \le 2 \sum_{k=1}^{K} \nu_k \left( e^2 2^{k-1} + 4 P_k\log(\sqrt{2}P_k/\beta) + 14P_k\right)$, and
from Eq.~\eqref{eq_nuk}, $\nu_k$ decreases monotonically as $k$ increases.
Therefore, to obtain an upper bound on $N$ under the constraint $P_K = \min\{2^{K-1}, 2^{\lfloor \log_2 P \rfloor}\}$, we choose $T_k$ and $P_k$ as follows:
\begin{align}
    \label{eq_Tk_Pk}
    T_k &= \begin{cases}
            2^{k-1} & ( k \in \{ 1, 2, ..., K_T \} ), \\
            2^{K_T-1} & ( k \in \{K_T + 1, K_T + 2, ..., K \} ),
          \end{cases}\nonumber \\
    P_k &=  \begin{cases}
                1 & ( k \in \{ 1, 2, ..., K_T \} ), \\
                2^{k-K_T} & ( k \in \{K_T + 1, K_T + 2, ..., K \} ),
              \end{cases}
\end{align}
where $K_T := \max\{1, K - \lfloor \log_2 P \rfloor\}$.
Substituting $T_k$, $P_k$, and $L_k$ from Eq.~\eqref{eq_Lk}, as well as $\nu_k$ from Eq.~\eqref{eq_nuk}, into $N$, we obtain
\begin{align}
    \label{eq_order_N}
    N &= 2 \sum_{k=1}^{K} \nu_k (L_k + 2) P_k  \nonumber \\
    &= 2 \sum_{k=1}^{K} \nu_k (L_k + 2) \cfrac{2^{k-1}}{T_k}\nonumber \\
    &\le 2 \sum_{k=1}^{K_T} \left\{ \gamma(K-k) + 2 \right\} \left\{ e^2 2^{k-1} + 4\log(\sqrt{2}/\beta)  + 14 \right\} \nonumber \\
    &\qquad + 2 \sum_{k=K_T+1}^{K} \left\{ \gamma(K-k) + 2 \right\} 2^{k-K_T} \left\{ e^2 2^{K_T-1} + 4\log{(\sqrt{2}\times2^{k-K_T}/\beta)} + 14 \right\} \nonumber \\
    &\lesssim 2^K + 2^{K-K_T}(K - K_T) + {\rm poly}(K) \nonumber \\
    &= \mathcal{O}\left( \cfrac{1}{\varepsilon} + P \log P \right),
\end{align}
where $\gamma := \frac{\log 6}{2(\sqrt{6}/8-\beta)^2}$, and we note that $\max_k P_k\leq P$ holds.

Finally, we consider the maximum depth and number of qubits. 
Since Eq.~\eqref{eq_Tk_Pk} shows that $T_k$ and $P_k$ attain their largest values at $k=K$, it suffices to evaluate depth and number of qubits at $k=K$.
From $T_K = \max \{1, 2^{K-\lfloor\log_2{P}\rfloor - 1}\}$, we have $L_K \le e^2 T_K + 4\log(\sqrt{2}P_K/\beta) + 12$, and since $\log P_K \le \log P$, it follows that $L_K = \mathcal{O}(1/(\varepsilon P) + \log P)$.
Therefore, the depth of $V_{\varphi,T_K}$ is $\mathcal{O}(L_K)=\mathcal{O}(1/(\varepsilon P) + \log P)$.
In addition, $\ket{{\rm GHZ}_{P}}$ can be constructed from $\ket{0}^{\otimes P}$ with the $\log P$-depth circuit~\cite{cruz2019efficient, mooney2021generation}.
Consequently, the total depth of the circuit is $\mathcal{O}(1/(\varepsilon P) + \log P)$.
Finally, at most $P$ instances of an $(n + 1)$-qubit system are arranged in parallel, thus the maximum number of qubits is $P(n+1)$. 
\hfill\qedsymbol

\section{Details of numerical experiment}
\label{sec_detail_numerical_experiment}
\subsection{Details of the numerical experiments for query and depth evaluation}
\label{sec_detail_numerical_experiment_query_complexity}
Here we provide details of the numerical experiment setup for the query complexity evaluation in the main text Fig.~\ref{fig_graph_query_comp}.
We set $n=2$ and estimated RMSE over 100 trials for $a \in \{ 0, \sin^2{(\pi/8)} \}$ and $K \in \{1, 2, ..., 9\}$.
As for the choice of $P_k$ and $T_k$, we considered the two cases: 
(i)~\textit{Full parallel:} fix $T_k = 1~\forall k$ and set $P_k=2^{k-1}$, and 
(ii)~\textit{Full sequential:} fix $P_k = 1~\forall k$ and set $T_k=2^{k-1}$. 
In case (i), we used Qiskit~\cite{qiskit2024} for quantum circuit simulation and a Python library~\cite{chao2020finding, chao2020finding_github, chuang2021pyqsp_github, martyn2021grand} to compute the QSP hyperparameters. 
In case (ii), the estimation was performed by sampling from the measurement distribution that assumes ideal operator transformations (i.e., $V_{\varphi} = \widetilde{V}_{\varphi}$). 
In both cases, we chose the measurement schedule $\nu_k$ as the RPE-optimized schedule $\nu_k = \lfloor 4.0835(K - k) + \nu_K\rceil$ with $\nu_K \in \{7, 18\}$~\cite{belliardo2020achieving}. 
Although this schedule is optimized under the assumption that the bias $\beta$ in the measurement probability is zero, we chose $L_k$ sufficiently large to make the effect of this bias negligible; in particular, we chose $L_k$ such that $|\beta_{r,k}| \le 0.05$.
Details of the $L_k$ setting are provided in SM Sec.~\ref{sec_experiment_oc_error}.

\subsection{Numerical evaluation of the approximation error in $V_{\varphi, T}$ and its impact on estimation}
\label{sec_experiment_oc_error}
In this section, we present results that show how to choose $L_k$ so that the effect of the approximation error in $V_{\varphi, T_k}$ on estimating $a$ is negligible.

First, we examined the relationship between the measurement probability bias $\beta$ and the query complexity.
Figure~\ref{fig_graph_query_comp_beta} shows the query complexity obtained from numerical experiments for various values of $\beta \in \{0.00, 0.05, 0.10, 0.15, 0.20, 0.25, 0.30 \}$.
In this experiment, we assumed $\beta_{+, k} = \beta_{i, k} = \beta$.
To compute $N$, we set $L_k = 1$ for all $k$.
We performed the estimation by sampling measurement outcomes according to the probabilities $p_{+,k} = (1 + \cos{M_k \varphi})/2 + \beta$ and $p_{i,k} = (1 + \sin{M_k \varphi}) / 2 + \beta$.
We evaluated $\varepsilon$ by performing 100 estimation trials for each $a \in \{0.00, 0.01, ..., 1.00\}$, and computed the average ($\varepsilon_{\rm avg}$) and maximum ($\varepsilon_{\rm max}$) values of $\varepsilon$.
All other parameters were set as in Sec.~\ref{sec_detail_numerical_experiment_query_complexity}, using $\nu_k = \lfloor 4.0835(K - k) + \nu_K\rceil$, $\nu_K \in \{7, 18\}$, and $K \in \{1, 2, ..., 9 \}$.
Based on the result in Fig.~\ref{fig_graph_query_comp_beta}, when  $\beta = 0.05$, the estimation accuracy is comparable to that achieved when $\beta = 0$, even though the settings of $\nu_k$ and $\nu_K$ are the same (i.e., the bias is not taken into account when configuring these parameters).
\begin{figure}[H]
    \centering
    \includegraphics[width=\columnwidth]{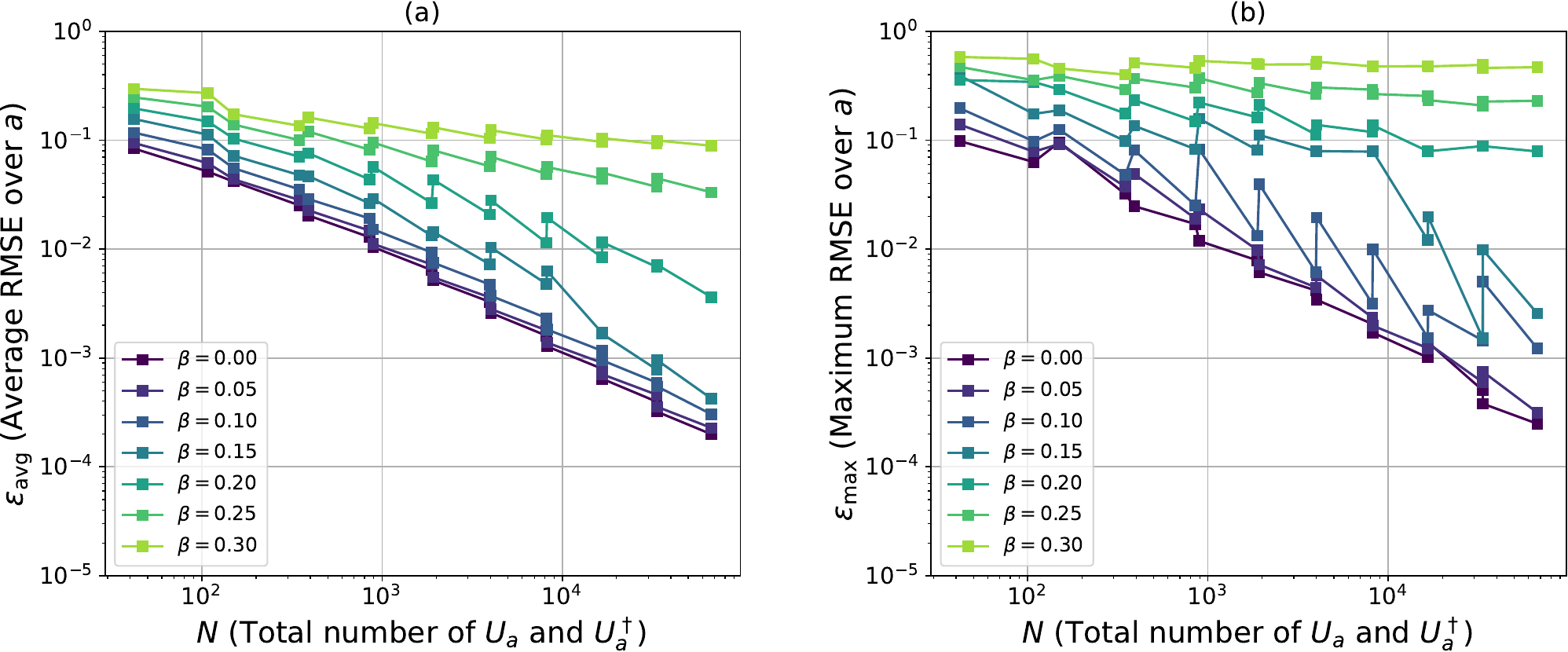}
    \caption{
        The relationship between the number of queries to $U_a$ and $U_a^{\dagger}$, and (a) the average value and (b) the maximum value of $\varepsilon$ over $a$, as the parameter $\beta$ varies. 
    }
    \label{fig_graph_query_comp_beta}
\end{figure}

Next, we considered the $L$ required to ensure $\beta \le 0.05$ in the full parallel case  (i.e. $T = 1$, $P = 2^{K-1}$).
We numerically evaluated the relationships between $L$ and $\beta$.
In this experiment, we fixed $T=1$, and evaluated $|\beta_{+, k}|$ and $|\beta_{i, k}|$ for $L \in \{ 4, 6, ..., 24 \}$ and $K \in \{1, 2, ..., 9 \}$.
$\beta_{+, k}$ and $\beta_{i, k}$ were computed via quantum circuit simulation using Qiskit~\cite{qiskit2024}, where the measurement probabilities $p_{+,k}$ and $p_{i,k}$ were estimated from 100000-shot measurements.
We calculated the angle sequence $\vec{\xi}$ using the Python library~\cite{chao2020finding, chao2020finding_github, chuang2021pyqsp_github, martyn2021grand}, as in the experiment described in the main text.
In Fig.~\ref{fig_graph_QSP-error}, $|\beta_{+, k}|$ and $|\beta_{i, k}|$ were computed as the maximum values over $a \in \{0.00, 0.01,  ..., 1.00\}$.
As shown in Fig.~\ref{fig_graph_QSP-error}, $|\beta_{+, k}|$ and $|\beta_{i, k}|$ decrease exponentially as $L$ increases for sufficiently large $L$.
This observation is consistent with the behavior predicted by Lemma~\ref{lm_operator_transformation} and the inequality $|\beta_{r, k}| \le \sqrt{2} P_k \max_j \big\lVert (V_{\varphi, T_k} - \widetilde{V}_{\varphi, T_k}) \ket{j}_b \ket{0}^{\otimes n}_s \big\rVert$ stated in the proof of Theorem~\ref{thm_PAE}.
Figure~\ref{fig_graph_QSP-error} also indicates that the condition $|\beta_{+,k}| \le 0.05$ and $|\beta_{i,k}| \le 0.05$, required to achieve accuracy comparable to the $\beta=0$ case in Fig.~\ref{fig_graph_query_comp_beta}, can be satisfied by an appropriate choice of $\{L_k\}$.
Specifically, for projective measurements in the basis $\{\ket{\pm_{P_k}}_b := (\ket{0}^{\otimes P_k}_b \pm \ket{1}^{\otimes P_k}_b)/\sqrt{2}\}$, we set $(L_1,\ldots,L_9)=(10,12,12,16,16,18,20,20,22)$, whereas for projective measurements in the basis $\{\ket{\pm i_{P_k}}_b := (\ket{0}^{\otimes P_k}_b \pm i\ket{1}^{\otimes P_k}_b)/\sqrt{2}\}$, we set $(L_1,\ldots,L_9)=(12,14,14,14,16,18,20,20,22)$.

\begin{figure}[H]
    \centering
    \includegraphics[width=\columnwidth]{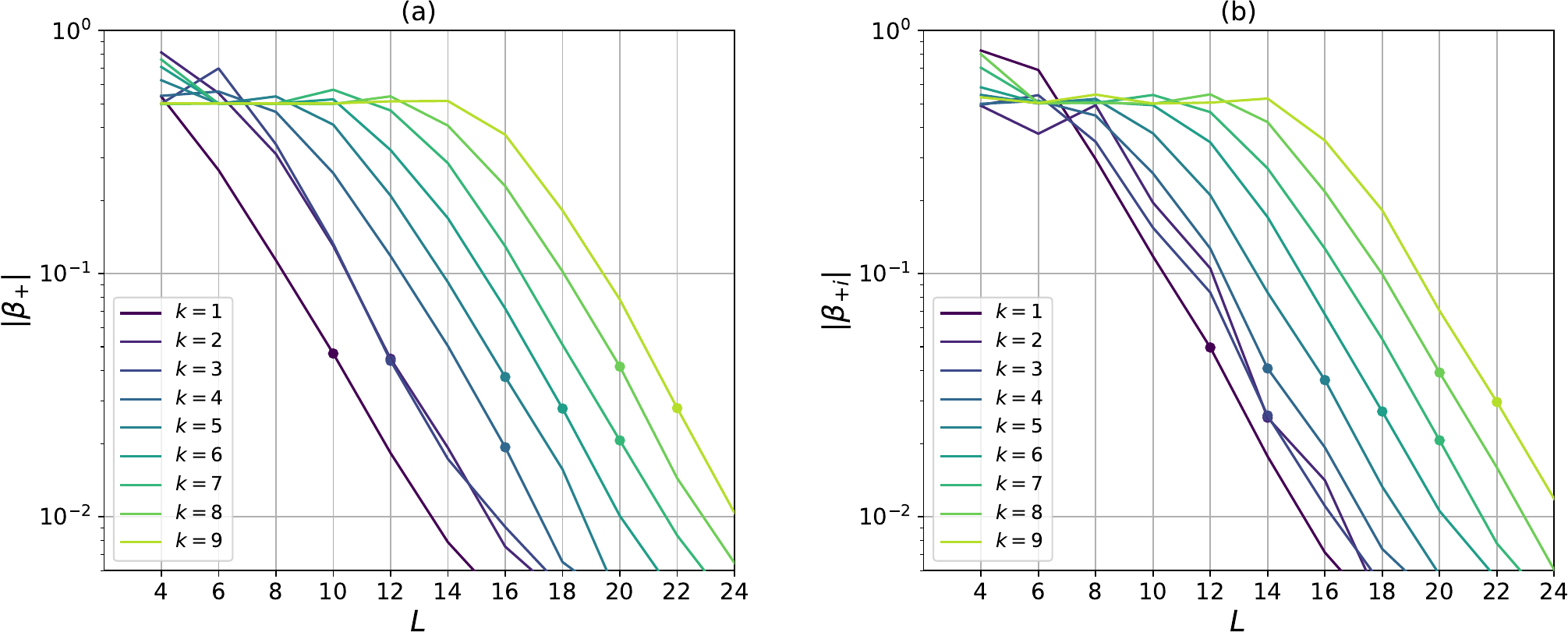}
    \caption{
        The relationship of $L \in \{ 2, 4, ..., 24 \}$ with (a) $|\beta_{+, k}|$ and (b) $|\beta_{i, k}|$, for $k \in \{1, 2, ..., 9\}$.
        Dots indicate the data points with the smallest $L$ that satisfy $|\beta_{+, k}| \le 0.05$ and $|\beta_{i, k}|  \le 0.05$ for each $k$.
    }
    \label{fig_graph_QSP-error}
\end{figure}

We also evaluate $L$ required to ensure $\beta \le 0.05$ in the full sequential case  (i.e. $P = 1$, $T = 2^{K-1}$).
According to Eq.~\eqref{eq_Vph_norm-error_1} and the inequality $|\beta_{r, k}| \le \sqrt{2} P_k \max_j \big\lVert (V_{\varphi, T_k} - \widetilde{V}_{\varphi, T_k}) \ket{j}_b \ket{0}^{\otimes n}_s \big\rVert$, the condition $\beta \le 0.05$ is fulfilled if $8\delta + \sqrt{16\delta - 64 \delta^2} \le 0.025$, which holds when $\delta < 3.813 \times 10^{-5}$.
From Eq.~\eqref{eq_delta_1}, this constraint on $\delta$ leads to the following condition on $L$:
\begin{align}
    \label{eq_T-L}
    \cfrac{4T^{L/2 + 1}}{2^{L/2 + 1}(L/2 + 1)!} < 3.813 \times 10^{-5}.
\end{align}
Figure~\ref{fig/graph_T-L.pdf} illustrates the $T$-$L$ relationship obtained by replacing the inequality sign '$<$' in Eq.~\eqref{eq_T-L} with the equality.
Based on this result, we use $(L_1, L_2, L_3, L_4) = (10, 14, 22, 34)$ in the numerical experiments described in the main text for the full sequential case.
In addition, as shown in Fig.~\ref{fig/graph_T-L.pdf}, a linear fit for $T \ge 10$ yields $L = 2.72 T + 13.64$.
Therefore, we set $L_k = 2\left\lceil (2.72 T_k + 13.64)/2 \right\rceil$ for $k \ge 5$.

\begin{figure}[H]
    \centering
    \includegraphics[width=0.6\columnwidth]{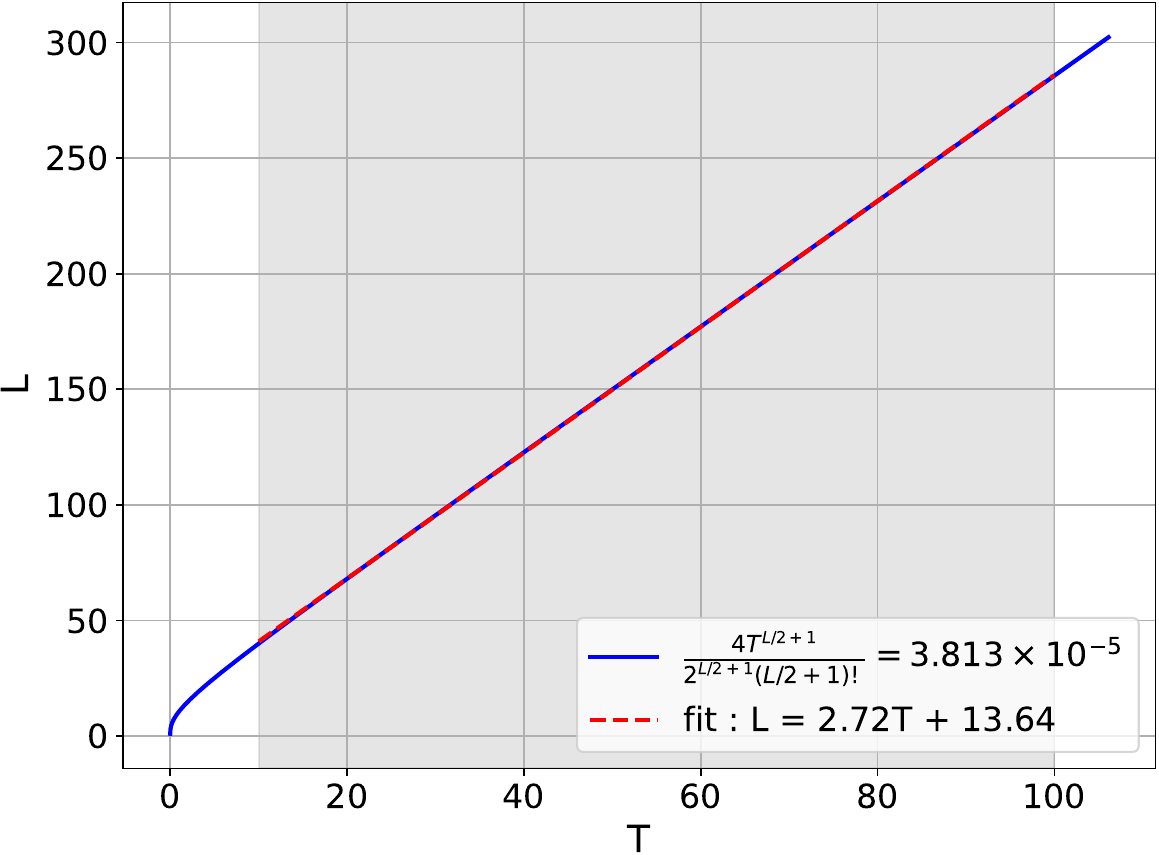}
    \caption{
        $T$-$L$ relationship derived by replacing the inequality in Eq.~\eqref{eq_T-L} with the equality.
        The blue line represents the resulting $T$-$L$ curve, while the red dashed line shows the linear fitting result for $10 \le T \le 100$.
    }
    \label{fig/graph_T-L.pdf}
\end{figure}

\section{The lower bound for parallel approximate counting}
\label{sec_parallel_adv_bound}

The main goal of this section is to prove Eq.~\eqref{Q in App} and its simplification 
Eq.~\eqref{eq:apdx_advlb_nontrivial}. 
For this purpose, we first give one of the basic results in quantum adversary method in Section~\ref{sec_parallel_adv_bound_thm}, which can be used for evaluating  the lower bound of the parallel query complexity (namely the minimal number (or depth) of parallel queries achieving the task).

\setcounter{theorem}{2}
\subsection{Parallel adversary lower bound for multi-valued functions}
\label{sec_parallel_adv_bound_thm}

\begin{theorem}[Parallel combinatorial adversary lower bound for multi-valued functions~\cite{burchard2019lower,jeffery2017optimal}]
    \label{thm_parallel_adv_bound}
    Let $\mathcal{X}$ and $\mathcal{Y}$ be sets of bit strings to a multi-valued function $F$ such that $F(x) \cap F(y) = \emptyset ~~ \forall x\in\mathcal{X},~\forall y\in\mathcal{Y}$, and let $P$ be a positive integer.
    For a relation $R \subseteq \mathcal{X} \times \mathcal{Y}$, we define $R^{i_1,\ldots,i_P} = \{ (x, y) \in R: \exists j \in \{1,\ldots,P\} ~~ {\rm s.t.} ~ x_{i_j} \ne y_{i_j}\}$, where $x_i\in \{0, 1\}$ denotes the $i$-th bit of $x$. Let us define $h,h',\ell,\ell'$ as
    \begin{gather}
        h := \min_{x \in \mathcal{X}} \left| \left\{ y\in\mathcal{Y}: (x, y) \in R \right\} \right|, ~~ h' := \min_{y \in \mathcal{Y}} \left| \left\{ x\in\mathcal{X}: (x, y) \in R \right\} \right|, \nonumber \\
        \ell := \max_{i_1,\ldots,i_P}\max_{x \in \mathcal{X}} \left| \left\{ y\in\mathcal{Y}: (x, y) \in R^{i_1,\ldots,i_P} \right\} \right|,~~\ell' := \max_{i_1,\ldots,i_P}\max_{y \in \mathcal{Y}} \left| \left\{ x\in\mathcal{X}: (x, y) \in R^{i_1,\ldots,i_P} \right\} \right|. \nonumber
    \end{gather}
    Then, for any quantum algorithm that computes an element of $F$ with high probability, the $P$-parallel query complexity is $
        \Omega\left(\sqrt{(hh')/(\ell \ell')}\right)$.
\end{theorem}

We provide a full proof of this theorem for completeness. 
According to the standard quantum adversary method~\cite{de2019quantum}, we introduce a binary oracle $O_x:\ket{i,b}\mapsto \ket{i,b\oplus x_i}$ ($b\in \{0,1\}$) for a bit string $x=x_1x_2...x_{\mathcal{N}}\in \{0,1\}^{\mathcal{N}}$. 
Then, the final quantum state $\ket{\psi_x^{T}}$ of any quantum algorithm with $T$ queries to $P$-parallel oracle $O_x^{\otimes P}$ can be described by
\begin{equation}
        \ket{\psi_x^{T}}=V_T(O_{x}^{\otimes P}\otimes \bm{1}_a)\cdots V_1(O_{x}^{\otimes P}\otimes \bm{1}_a)V_0\ket{0}^{\otimes P}\ket{0}_{\rm a}.
\end{equation}
Here, the number of the ``workspace'' ancilla qubits $\ket{0}_{\rm a}$ is arbitrary finite.
The unitary gates $V_1,...,V_T$ are independent of $x$. 
We also denote the quantum state after $t$ queries as $\ket{\psi_x^t}$. 
A rough idea of adversary method is to derive the necessary $T$ such that we can distinguish $\ket{\psi_x^T}$ and an adversary $\ket{\psi_y^T}$.

\begin{proof}[Proof of Theorem~\ref{thm_parallel_adv_bound}]
    First of all, we can bound the number $|R|$ of elements in $R$ as
    \begin{equation}
        |R|:=\sum_{x\in \mathcal{X}}\sum_{y\in \mathcal{Y}} \chi_R(x,y)=\sum_{x\in \mathcal{X}}|\{y\in \mathcal{Y}:(x,y)\in R\}|\geq h|\mathcal{X}|.
    \end{equation}
    Here, $\chi_R(x,y)$ is the indicator function such that if $(x,y)\in R$, then $\chi_R(x,y)=1$; otherwise, $\chi_R(x,y)=0$.
    Similarly, we obtain $|R|\geq h'|\mathcal{Y}|$. 
    As in the proof of the adversary method, we define progress $\Delta_t$ at $t$ as
    \begin{equation}
        \Delta_t:=\sum_{(x,y)\in R}|\langle \psi_x^t|\psi_y^t\rangle|,
    \end{equation}
    and evaluate the possible largest difference between $\Delta_{t}$ and $\Delta_{t+1}$.
    Let us define $I=(i_1,...,i_P)\in \{1,2,...,\mathcal{N}\}^P$.
    The quantum state $\ket{\psi_x^t}$ can be expanded as \begin{equation}
        \ket{\psi_x^t}=\sum_{I\in \{1,2,...,\mathcal{N}\}^P}\beta_I^{(x,t)}\ket{I}\ket{\phi_{I}^{(x,t)}},
    \end{equation} with some complex amplitudes $\beta^{(x,t)}_{I}$ and ancillary quantum states $\ket{\phi^{(x,t)}_{I}}$.
    Here, we note that there exists a unitary $U_{x,I}$ such that $O_x^{\otimes P}\ket{I}\ket{b_1,...,b_P}=\ket{I}U_{x,I}\ket{b_1,...,b_P}$; its action is the bit permutation $U_{x,I}\ket{b_1,...,b_P}=\ket{b_1\oplus x_{i_1},...,b_P\oplus x_{i_P}}$.
    Thus, a single $P$-parallel query changes the state $\ket{\psi_x^t}$ into
    \begin{equation}
        (O_x^{\otimes P}\otimes \bm{1}_a)\ket{\psi_x^t}=\sum_{I }\beta_I^{(x,t)}\ket{I}(U_{x,I}\otimes \bm{1}_{\rm a})\ket{\phi_{I}^{(x,t)}},
    \end{equation}
    which yields
    \begin{align}
        \left|\bra{\psi_x^{t+1}}\psi_y^{t+1}\rangle-\bra{\psi_x^{t}}\psi_y^{t}\rangle\right|&\leq  \sum_{I}\left|\overline{\beta^{(x,t)}_{I}}\beta^{(y,t)}_{I}\right|\left|\bra{\phi_{I}^{(x,t)}}\left(U_{x,I}^{\dagger}U_{y,I}\otimes \bm{1}_{\rm a}-\bm{1}\right)\ket{\phi_{I}^{(y,t)}}\right|\leq  2\sum_{I:x_I\neq y_I}|{\beta^{(x,t)}_{I}}||\beta^{(y,t)}_{I}|.
    \end{align}
    In the final inequality, we used the fact that if $x_{i_1}x_{i_2}...x_{i_P}=:x_I=y_I$, then $U_{x,I}^\dagger U_{y,I}$ becomes the identity; otherwise, $\|U_{x,I}^\dagger U_{y,I}-\bm{1}\|\leq 2$.
    By using this, we can bound the difference between $\Delta_t$ and $\Delta_{t+1}$: for any positive value $\gamma>0$,
    \begin{align}\label{eq:multi_val_F_parallel_diffprogress}
        \Delta_t-\Delta_{t+1}&= \sum_{(x,y)\in R}\left(|\langle \psi_x^{t}|\psi_y^{t}\rangle|-|\langle \psi_x^{t+1}|\psi_y^{t+1}\rangle|\right)\leq \sum_{(x,y)\in R} \sum_{I:x_I\neq y_I}2|{\beta^{(x,t)}_{I}}||\beta^{(y,t)}_{I}|\notag\\
        &\leq \sum_{(x,y)\in R} \sum_{I:x_I\neq y_I}\left(\gamma|{\beta^{(x,t)}_{I}}|^2+\frac{1}{\gamma}|\beta^{(y,t)}_{I}|^2\right),
    \end{align}
    where we used the AM-GM inequality for positive values $\gamma|{\beta^{(x,t)}_{I}}|^2$ and $\frac{1}{\gamma}|\beta^{(y,t)}_{I}|^2$.

    Hereafter, we evaluate the sums in Eq.~\eqref{eq:multi_val_F_parallel_diffprogress}.
    \begin{align}
        &\sum_{(x,y)\in R}\sum_{I:x_I\neq y_I}|{\beta^{(x,t)}_{I}}|^2=\sum_{I}\sum_{(x,y)\in R}\chi_{x_I\neq y_I}(I)|{\beta^{(x,t)}_{I}}|^2=
        \sum_{I}\sum_{(x,y)\in R:x_I\neq y_I}|{\beta^{(x,t)}_{I}}|^2=\sum_{I}\sum_{(x,y)\in R^{I}}|{\beta^{(x,t)}_{I}}|^2\notag\\
        &~~~~~=\sum_{I}\sum_{x,y}\chi_{R^{I}}(x,y)|{\beta^{(x,t)}_{I}}|^2=\sum_{I}\sum_{x\in \mathcal{X}}\left(|{\beta^{(x,t)}_{I}}|^2\sum_{y\in \mathcal{Y}}\chi_{R^{I}}(x,y)\right)\leq\sum_{I}\sum_{x\in \mathcal{X}}\ell|{\beta^{(x,t)}_{I}}|^2=\ell|\mathcal{X}|,
    \end{align}
    where we shorten $R^{i_1,...,i_P}$ to $R^{I}$.
    Similarly,
    \begin{align}
        &\sum_{(x,y)\in R}\sum_{I:x_I\neq y_I}|{\beta^{(y,t)}_{I}}|^2=\sum_{I}\sum_{(x,y)\in R}\chi_{x_I\neq y_I}(I)|{\beta^{(y,t)}_{I}}|^2=
        \sum_{I}\sum_{(x,y)\in R:x_I\neq y_I}|{\beta^{(y,t)}_{I}}|^2=\sum_{I}\sum_{(x,y)\in R^{I}}|{\beta^{(y,t)}_{I}}|^2\notag\\
        &~~~~~=\sum_{I}\sum_{x,y}\chi_{R^{I}}(x,y)|{\beta^{(y,t)}_{I}}|^2=\sum_{I}\sum_{y\in \mathcal{Y}}\left(|{\beta^{(y,t)}_{I}}|^2\sum_{x\in \mathcal{X}}\chi_{R^{I}}(x,y)\right)\leq\sum_{I}\sum_{y\in \mathcal{Y}}\ell'|{\beta^{(y,t)}_{I}}|^2=\ell'|\mathcal{Y}|.
    \end{align}
    These evaluations yield
    \begin{equation}
        \Delta_t-\Delta_{t+1}\leq \gamma \ell|\mathcal{X}|+\frac{1}{\gamma}\ell'|\mathcal{Y}|
        \leq \gamma \frac{\ell}{h}|R|+\frac{1}{\gamma}\frac{\ell'}{h'}|R|,
    \end{equation}
    where we used $|R|\geq h|\mathcal{X}|$ and $|R|\geq h'|\mathcal{Y}|$.
    Since this inequality holds for any $\gamma>0$, we minimize the upper bound by taking $\gamma=\sqrt{h\ell'/(h'\ell)}$.
    As a result,
    \begin{equation}
        \Delta_t-\Delta_{t+1}\leq 2\sqrt{\frac{\ell\ell'}{hh'}}|R|.
    \end{equation}

    Computing an element of the set $F(x)$ with a success probability at least $1-\delta$ requires $\bra{\psi_x^T}\Pi_{F(x)}\ket{\psi_x^T}\geq 1-\delta$ for an orthogonal projector $\Pi_{F(x)}$ onto the subspace corresponding to the possible values of $F(x)$.
    When $F(x)\cap F(y)=\emptyset$, it implies $\Pi_{F(x)}\Pi_{F(y)}=0$. 
    Thus, 
    \begin{align}
        \sqrt{1-|\bra{\psi_x^{T}}\psi_y^{T}\rangle|^2}&\geq \frac{1}{2}\|\ket{\psi_x^{T}}\bra{\psi_x^{T}}-\ket{\psi_y^{T}}\bra{\psi_y^{T}}\|_{1}\notag\\
        &\geq {\rm tr}[\Pi_{F(x)}(\ket{\psi_x^{T}}\bra{\psi_x^{T}}-\ket{\psi_y^{T}}\bra{\psi_y^{T}})]\geq 1-\delta- {\rm tr}[\Pi_{F(x)}\ket{\psi_y^{T}}\bra{\psi_y^{T}}]\notag\\
        &\geq 1-2\delta,
    \end{align}
    and 
    $|\bra{\psi_x^{T}}\psi_y^{T}\rangle|\leq 2\sqrt{\delta(1-\delta)}\equiv c$,
    where we used
    \begin{equation}
        0\leq{\rm tr}[\Pi_{F(x)}\ket{\psi_y^{T}}\bra{\psi_y^{T}}]={\rm tr}[(1-\Pi_{\overline{F(x)\cup F(y)}}-\Pi_{F(y)})\ket{\psi_y^{T}}\bra{\psi_y^{T}}]\leq \delta.
    \end{equation}
    Therefore,
    \begin{equation}
        (1-c)|R|\leq |R|-\Delta_T=\Delta_0-\Delta_T= \sum_{t=0}^{T-1}\left(\Delta_t-\Delta_{t+1}\right)\leq 2\sqrt{\frac{\ell\ell'}{hh'}}|R|T,
    \end{equation}
    and we finally arrive at $T\geq \frac{1-c}{2}\sqrt{\frac{hh'}{\ell\ell'}}$, which completes the proof.
\end{proof}

\subsection{Remarks on~Ref.~\cite{burchard2019lower}}
\label{sec_remarks_buchardpaper}

We now see that the lower bound derived in Ref.~\cite{burchard2019lower} has an error in its derivation. 
Theorem 3 in Ref.~\cite{burchard2019lower} argues that for an approximate counting problem with a relative error $\varepsilon_{\rm rel}$, any quantum algorithm with $P$-parallel queries has query depth $\Omega(\varepsilon^{-1}_{\rm rel}\cdot \sqrt{N_d/(PN_t)})$, where $N_t$ ($\neq 0$) denotes the number of marked items in a size-$N_d$ database. 
This was derived from the above Theorem~\ref{thm_parallel_adv_bound} (Theorem 2 in Ref.~\cite{burchard2019lower}) by calculating the factors $h,h',\ell,\ell'$ in the approximate counting problem.
However, we have identified that the evaluation of the parameter $\ell$ has been underestimated, thereby leading to an overly strong lower bound.

More precisely, Ref.~\cite{burchard2019lower} considers the following setup for the parallel quantum adversary method:
\begin{itemize}
            \item $\varepsilon_{\rm rel}$: a relative error $\varepsilon_{\rm rel}\in (0,1)$.
            \item $F(x)$: a multi-valued function for approximate counting satisfying $F(x)= \{z\in\mathbb{R}:|z-|x|_1|\leq \varepsilon_{\rm rel}|x|_1/3\}$, where $|\cdot|_1$ denotes the $\ell_1$ norm.
            \item $\mathcal{X}$: a set of length-$N_d$ bit strings $x\in\{0, 1 \}^{N_d}$ that have exactly $N_t$ ones.
            \item $\mathcal{Y}$: a set of length-$N_d$ bit strings $y\in\{0, 1 \}^{N_d}$ that have $N_t+\lceil \varepsilon_{\rm rel}N_t\rceil$ ($\leq N_d$) ones.
            \item $R$: a set of the pair $(x,y)\in \mathcal{X}\times \mathcal{Y}$, defined as $R := \{(x,y)\in \mathcal{X}\times \mathcal{Y}: x \le y\}$, where $x \le y$ means that for every index $i$, if $x_i = 1$, then $y_i = 1$.
            \item $R^{i_1...i_P}$: a subset of $R$, defined as $R^{i_1...i_P} := \{(x,y) \in R: \exists j\in \{1,...,P\}~{\rm s.t.} \: x_{i_j} \ne y_{i_j} \}$, where $i_j \in \{1, ..., N_d\}$ denotes the coordinate of a bit string.
\end{itemize}
Note that while
the original paper defines $\mathcal{Y}:=\{y:|y|_1=N_t+\varepsilon_{\rm rel }N_t\}$, we use the above definition with the ceil function to well-define $\mathcal{Y}$. 
Also, we need to take $\varepsilon_{\rm rel}/3$ in $F(x)$ as above, instead of $\varepsilon_{\rm rel}/2$ in the original paper; otherwise, the condition $F(x)\cap F(y)=\emptyset$ may be failed when $\varepsilon_{\rm rel}<1$.
Under these definitions, in~\cite{burchard2019lower}, 
the previous lower bound $\Omega(\varepsilon^{-1}_{\rm rel}\cdot \sqrt{N_d/(PN_t)})$ was derived by calculating $\ell$ as follows;
\begin{equation}\label{eq:wrong_evaluation}
    \ell= \binom{N_d-N_t-1}{\varepsilon_{\rm rel} N_t-1}~~~\mbox[\rm wrong]
\end{equation}
for (at least) $\varepsilon_{\rm rel}> 1/N_t$.
However, we have found that the factor $\ell$ can become larger than this value.
Our careful evaluation clarifies
\begin{align}\label{reply_corrected_l}
    \ell &= \left\{
            \begin{array}{ll}
            \dbinom{N_d - N_t}{\lceil \varepsilon_{\rm rel}N_t\rceil} - \dbinom{N_d - N_t - P}{\lceil \varepsilon_{\rm rel}N_t\rceil} & ( P \leq N_d - N_t -\lceil \varepsilon_{\rm rel}N_t\rceil) \\[10pt]
            \dbinom{N_d - N_t}{\lceil \varepsilon_{\rm rel}N_t\rceil} & (\mbox{otherwise}).
            \end{array}
            \right.
\end{align}
Indeed, even if $P=2$, Eq.~\eqref{eq:wrong_evaluation} and Eq.~\eqref{reply_corrected_l} are not the same: 
\begin{equation}
     \binom{N_d-N_t-1}{\varepsilon_{\rm rel} N_t-1}=\dbinom{N_d - N_t}{\varepsilon_{\rm rel}N_t} - \dbinom{N_d - N_t - 1}{\varepsilon_{\rm rel}N_t}<\dbinom{N_d - N_t}{\varepsilon_{\rm rel}N_t} - \dbinom{N_d - N_t - P}{\varepsilon_{\rm rel}N_t}
\end{equation}
when $\varepsilon_{\rm rel}N_t\in \mathbb{Z}$.
In the following, we derive the corrected adversary lower bound together with the derivation of Eq.~\eqref{reply_corrected_l}.

\subsection{Proof of the bounds presented in Appendix~\ref{sec_parallel_query_counting}}
\label{sm_sec:proofapdxc}
We now provide the proof of the lower and upper bounds presented in Appendix~\ref{sec_parallel_query_counting}.
These bounds are summarized in the following lemma.
\begin{lemma}[Lower bound for parallel approximate counting]\label{lem:corrected_ADVLB_approxcount}
    Let us consider a size-$N_d$ database with $N_t$ marked items such that $N_t+\lceil \varepsilon_{\rm rel}N_t\rceil\leq N_d$ for a relative error $\varepsilon_{\rm rel}\in  (0,1)$.
    (This is the same assumption as in Ref.~\cite{burchard2019lower}.)
    Then, for any quantum algorithm solving the approximate counting problem with high probability, the lower bound of the $P$-parallel query complexity is 
    \begin{equation}\label{eq:ADV_lb_corrected}
        \sqrt{\frac{hh'}{\ell\ell'}}=\left[1-\frac{\dbinom{N_d - N_t - P}{\lceil\varepsilon_{\rm rel}N_t\rceil}}{\dbinom{N_d - N_t}{\lceil\varepsilon_{\rm rel}N_t\rceil}}\right]^{-1/2}\left[1-\frac{\dbinom{N_t+\lceil\varepsilon_{\rm rel}N_t\rceil - P}{\lceil\varepsilon_{\rm rel}N_t\rceil}}{\dbinom{N_t+\lceil\varepsilon_{\rm rel}N_t\rceil}{\lceil\varepsilon_{\rm rel}N_t\rceil}}\right]^{-1/2}
    \end{equation}
    up to a constant factor, where we define $\binom{n}{r}=0$ if $n<r$.
    Furthermore, the following upper bound always holds:
    \begin{equation}\label{eq:consistent_with_PAE}
        \sqrt{\frac{hh'}{\ell\ell'}}=\mathcal{O}\left[{\rm PAE}^{P}\left(\frac{\varepsilon_{\rm rel}}{N_d/N_t}\right)\right],~~\mbox{where}~~{\rm PAE}^{P}(\varepsilon)=\mathcal{O}\left(\frac{1}{\varepsilon P}+\log P\right).
    \end{equation}
    For a nontrivial regime where all the binomial coefficients does not vanish, the $P$-parallel query complexity is  
    \begin{equation}
        \Omega\left(\frac{1}{P}\frac{N_t}{\lceil \varepsilon_{\rm rel}N_t\rceil}\sqrt{\frac{N_d-N_t(1+\varepsilon_{\rm rel})}{N_t}}\right).
    \end{equation}
\end{lemma}

\begin{proof}
    Let us consider the same setup described in Sec.~\ref{sec_remarks_buchardpaper}.
    We note that any quantum algorithm solving the approximate counting problem with a relative error $\varepsilon_{\rm rel}/3$ can compute an element of $F(x)$ with a high success probability.
    Therefore, it follows that the lower bound of the $P$-parallel query complexity of approximate counting is $\Omega(\sqrt{hh'/(\ell\ell')})$ by Theorem~\ref{thm_parallel_adv_bound} for the current setup.
    We then evaluate $(h,h',\ell,\ell')$ defined in Theorem~\ref{thm_parallel_adv_bound} as follows.
    For simplicity, we define $m:=\lceil \varepsilon_{\rm rel}N_t\rceil\geq 1$.
\begin{itemize}
    \item $h := \min_{x \in \mathcal{X}} \left| \left\{ y\in \mathcal{Y}: (x, y) \in R \right\} \right|$. \\
        Fix $x \in \mathcal{X}$ arbitrarily.
        Any $y \in \mathcal{Y}$ satisfying $(x, y)\in R$ is constructed by choosing $m$ additional 1's among the $N_d - N_t$ 0-positions of $x$.
        Thus, $\left| \left\{ y\in\mathcal{Y}: (x, y) \in R \right\} \right| = \binom{N_d - N_t}{m}$, which is independent of $x$.
        Therefore, 
        \begin{align}
            \label{eq_h_supp}
            h = \binom{N_d - N_t}{m}.
        \end{align}
    \item $h' := \min_{y \in \mathcal{Y}} \left| \left\{ x\in\mathcal{X}: (x, y) \in R \right\} \right|$. \\
        Fix $y \in \mathcal{Y}$ arbitrarily.
        An $x \in \mathcal{X}$ satisfies $(x, y)\in R$ iff $x$ is obtained by selecting $N_t$ 1's among the $N_t+m$ 1-positions of $y$.
        Thus, $\left| \left\{ x\in\mathcal{X}: (x, y) \in R \right\} \right| = \binom{N_t+m}{N_t}$, which is also independent of $y$.
        Therefore, 
        \begin{align}
            \label{eq_h'_supp}
            h' = \binom{N_t+m}{N_t}=\binom{N_t+m}{m}.
        \end{align}
    \item $\ell := \max_{i_1,\ldots,i_P}\max_{x \in \mathcal{X}} \left| \left\{ y\in\mathcal{Y}: (x, y) \in R^{i_1,\ldots,i_P} \right\} \right|$. \\
        Fix $x \in \mathcal{X}$ and a set of $P$-parallel query indices $I = \{i_1, \ldots, i_P\}$ arbitrarily.
        A pair $(x,y)\in R$ belongs to $R^I$ iff at least one of the different $m$ bits between $x$ and $y$ is in $I$.
        Equivalently, a pair $(x,y)\in R$ is not in $R^I$ iff all such $m$ bits are in $\{i:x_i=0\}\setminus I^{(x)}$, where $I^{(x)}:=I\cap \{i:x_i=0\}$.
        Thus, it is clear that the number of possible $y$ is maximized when all indices in $I$ are in $\{i:x_i=0\}$ and different.
        In this case, $|\{i:x_i=0\}\setminus I^{(x)}|=N_d-N_t-P$.
        If $N_d-N_t-P\geq m$, the number of $y$ such that $(x,y)$ is in $R$ but not in $R^I$ is equal to
        $\binom{N_d-N_t-P}{m}$,
        since we may choose all $m$ added positions from $|\{i:x_i=0\}\setminus I|=N_d-N_t-P$.
        If $N_d-N_t-P< m$, it is impossible to add all $m$ 1's to $x$ at positions $\{i:x_i=0\}\setminus I$; any pair $(x,y)\in R$ belongs to $R^I$.
        Therefore,
        \begin{align}
            \label{eq_ell_supp}
            \ell = \left\{
            \begin{array}{ll}
            h - \dbinom{N_d - N_t - P}{m} & ( P \leq N_d - N_t -m) \\
            h & (\mbox{otherwise}).
            \end{array}
            \right.
        \end{align}       
    \item $\ell' := \max_{i_1,\ldots,i_P}\max_{y \in \mathcal{Y}} \left| \left\{ x\in\mathcal{X}: (x, y) \in R^{i_1,\ldots,i_P} \right\} \right|$. \\
        Fix $y \in \mathcal{Y}$ and a set of $P$-parallel query indices $I = \{i_1, \ldots, i_P\}$ arbitrarily.
        According to the relation $R$, a disagreement $x_i \ne y_i$ can only occur when $y_i = 1$ and $x_i = 0$.
        Therefore, in order to maximize $\left| \left\{ x\in \mathcal{X}: (x, y) \in R^{I} \right\} \right|$, we choose as many indices of $I$ as possible from the 1-positions of $y$.
        If $P \le N_t$, $\left| \left\{ x\in\mathcal{X}: (x, y) \in R \setminus R^{I} \right\} \right|$ is equal to the number of ways to assign the $N_t-P$ 1's on the $N_t +m- P$ candidate indices after fixing $x_{i_1}=\cdots=x_{i_P}=1$.
        Thus, $\max_{I} \left| \left\{ x\in\mathcal{X}: (x, y) \in R^{I} \right\} \right| = h' - \binom{N_t+m - P}{N_t-P}$ in this case.
        Additionally, if $N_t < P$, we can choose $I$ to be a subset of the 1's positions of $y$ with $|I|=P$ or to be a set including all 1's positions of $y$, and no string $x\in\mathcal{X}$ can satisfy $x_i=1$ for all $i\in I$, hence all $x$ which satisfies $(x,y)\in R$ differs from $y$ on at least one queried index.
        Thus, $\left| \left\{ x\in\mathcal{X}: (x, y) \in R \setminus R^{I} \right\} \right| = 0$ and $\max_{I} \left| \left\{ x\in\mathcal{X}: (x, y) \in R^{I} \right\} \right| = h'$.
        In both cases, $\max_{I}\left|\{x\in\mathcal{X}:(x,y)\in R^{I}\}\right|$ is independent of $y$.
        Therefore, 
        \begin{align}
            \label{eq_ell'_supp}
            \ell' &= \left\{
            \begin{array}{ll}
            h' - \dbinom{N_t+m - P}{N_t-P} & ( P \le N_t ) \\
            h' & (\mbox{otherwise})
            \end{array}
            \right.= \left\{
            \begin{array}{ll}
            h' - \dbinom{N_t+m - P}{m} & ( P \le N_t ) \\
            h' & (\mbox{otherwise}).
            \end{array}
            \right.
        \end{align}
\end{itemize}
Thus, we complete the proof of Eq.~\eqref{eq:ADV_lb_corrected}.

To evaluate the upper and lower bounds of $\sqrt{hh'/(\ell\ell')}$, we here simplify the factors $\ell/h$ and $\ell'/h'$ for the nontrivial regime as 
\begin{align}
    \frac{\ell}{h}&\equiv 1-\frac{\dbinom{N_d - N_t - P}{m}}{\dbinom{N_d - N_t}{m}}=1-\prod_{i=0}^{m-1} \frac{N_d-N_t-P-i}{N_d-N_t-i}=1-\prod_{i=0}^{m-1} \left(1-\frac{P}{N_d-N_t-i}\right),
\end{align}
\begin{align}
    \frac{\ell'}{h'}&\equiv 1-\frac{\dbinom{N_t+m - P}{m}}{\dbinom{N_t+m}{m}}
    =1-\prod_{i=0}^{m-1}\frac{N_t+m-P-i}{N_t+m-i}=1-\prod_{i=0}^{m-1}\left(1-\frac{P}{N_t+m-i}\right).
\end{align}
Now, we use the following inequalities for any $x_i\in [0,1)$
\begin{equation}
    \frac{\sum_i x_i}{1+\sum_i x_i}\leq 1-\prod_{i=0}^{m-1}(1-x_i)\leq \sum_i x_i,
\end{equation}
which can be proved by mathematical induction.
Hence, when $N_d-N_t-P\geq m$ (equivalently, $N_d-N_t-m\geq P$), $P/(N_d-N_t-i)\in [0,1)$ holds for any $i=0,1,...,m-1$ and we have
\begin{equation}
    \left(1+\frac{N_d-N_t}{mP}\right)^{-1}\leq \frac{\ell}{h}\leq \frac{mP}{N_d-N_t-m+1}.
\end{equation}
Similarly, when $P\leq N_t$, $P/(N_t+m-i)\in [0,1)$ holds for any $i=0,1,...,m-1$ and we have
\begin{equation}
    \left(1+\frac{N_t+m}{mP}\right)^{-1}\leq \frac{\ell'}{h'}\leq \frac{mP}{N_t+1}.
\end{equation}
These evaluations immediately yield the lower bound of the $P$-parallel query complexity for the nontrivial regime 
\begin{equation}
    \Omega\left(\sqrt{\frac{hh'}{\ell\ell'}}\right)= \Omega\left(\frac{1}{P}\frac{N_t}{\lceil \varepsilon_{\rm rel}N_t\rceil}\sqrt{\frac{N_d-N_t(1+\varepsilon_{\rm rel})}{N_t}}\right).
\end{equation}

Finally, we confirm Eq.~\eqref{eq:consistent_with_PAE} in all the four regimes: (i) the nontrivial regime $P\leq N_d-N_t-m$ and $P\leq N_t$, (ii) $P\leq N_d-N_t-m$ and $P> N_t$, (iii) $P> N_d-N_t-m$ and $P\leq N_t$, and (iv) the trivial regime $P>N_d-N_t-m$ and $P> N_t$.
\begin{itemize}
    \item[(i)] the nontrivial regime $P\leq N_d-N_t-m$ and $P\leq N_t$
    \begin{align}
        \sqrt{\frac{hh'}{\ell\ell'}}&\leq \left(1+\frac{N_d-N_t}{mP}\right)^{1/2}\left(1+\frac{N_t+m}{mP}\right)^{1/2}\leq 1+\frac{N_d+m}{2mP}&&\text{(the AM-GM inequality)}\notag\\
        &\leq  1+\frac{N_d}{2\varepsilon_{\rm rel}PN_t}+\frac{1}{2P}=\mathcal{O}\left(\frac{1}{\varepsilon_{\rm rel}P}\frac{N_d}{N_t}+\log P\right).
    \end{align}

    \item[(ii)] $P\leq N_d-N_t-m$ and $P> N_t$
    \begin{align}
        \sqrt{\frac{hh'}{\ell\ell'}}\leq \left(1+\frac{N_d-N_t}{mP}\right)^{1/2}\leq 1+\sqrt{\frac{N_d-N_t}{mP}}\leq 1+\sqrt{\frac{N_d}{\varepsilon_{\rm rel}PN_t}}=\mathcal{O}\left(\frac{1}{\varepsilon_{\rm rel}P}\frac{N_d}{N_t}+\log P\right).
    \end{align}

    \item[(iii)] $P> N_d-N_t-m$ and $P\leq N_t$
    \begin{align}
        \sqrt{\frac{hh'}{\ell\ell'}}\leq \left(1+\frac{N_t+m}{mP}\right)^{1/2}\leq 1+\sqrt{\frac{N_t+m}{mP}}\leq 1+\mathcal{O}\left(\sqrt{\frac{N_t}{\varepsilon_{\rm rel}N_t P}}\right)=\mathcal{O}\left(\frac{1}{\varepsilon_{\rm rel}P}\frac{N_d}{N_t}+\log P\right).
    \end{align}
    \item[(iv)] the trivial regime $P>N_d-N_t-m$ and $P> N_t$
    \begin{equation}
        \sqrt{\frac{hh'}{\ell\ell'}}=1=\mathcal{O}\left(\frac{1}{\varepsilon_{\rm rel}P}\frac{N_d}{N_t}+\log P\right).
    \end{equation}
\end{itemize}
\end{proof}


\end{document}